\documentclass[journal,10pt]{IEEEtran}
\usepackage{graphicx,amsthm}
\usepackage [noadjust]{cite} %[nospace]
\usepackage{bm}  % bold math for greek symbols
\usepackage{amsmath}
\usepackage{amssymb}
\usepackage{enumerate}
\usepackage{stfloats}
\usepackage{cases}
\usepackage{epstopdf}
\usepackage{soul,color} % \ul  \hl
\usepackage{hyperref}
\usepackage[percent]{overpic}
\usepackage{tabularx}
\usepackage[table]{xcolor}
\usepackage{multirow}
\usepackage{booktabs}  % multicolumn tabular
\usepackage{tabu} % different font sizes in different rows of tabular

\usepackage{algorithm}
\usepackage{algpseudocode}% 
\newcommand\numberthis{\addtocounter{equation}{1}\tag{\theequation}}

\newcommand{\algstartblock}
{	
	\ifnum\value{algcounter}<10{
		\addtocounter{algblockspace}{12}
	}\else{
		\addtocounter{algblockspace}{12}
	}
	\fi
}
\newcommand{\algendblock}{
	\ifnum\value{algcounter}<10{
		\addtocounter{algblockspace}{-12}
	}\else{
		\addtocounter{algblockspace}{-12}
	}
	\fi
}

\newcommand{\alg}[1]{
	\ifnum\value{algcounter}<10{
		\begin{enumerate}[\arabic{algcounter}.\hspace{\numexpr\value{algblockspace}+0 pt}]
			\setlength\itemindent{5pt}
			\item \hspace{-6pt} #1
		\end{enumerate} \stepcounter{algcounter}
	}\else{
		\begin{enumerate}[\arabic{algcounter}.\hspace{\numexpr\value{algblockspace}+0 pt}]
			\setlength\itemindent{0pt}
			\item \hspace{-6pt} #1
		\end{enumerate} \stepcounter{algcounter}
	}
	\fi
}

\newcommand{\Iid}{\mathcal{I}_{i,j}^d}
\newcommand{\Ijc}{\mathcal{I}_{j,i}^c}
\newcommand{\Ik}{\mathcal{I}_{k,k'}}
\newcommand{\Izeroid}{\widetilde{\mathcal{I}}_{i,j}^d}
\newcommand{\Izerojc}{\widetilde{\mathcal{I}}_{j,i}^c}
\newcommand{\eps}{\varepsilon}

\newcommand{\M}{\mathcal{M}}

\newcommand{\argmax}{\operatornamewithlimits{argmax}}

\newcommand{\pbold}{\mathbf{p}}

\newcommand{\gammatarget}{\widehat{\gamma}}
\newcommand{\npower}{\sigma^2}

\newtheorem{lemma}{Lemma}

\newtheorem{theorem}{Theorem}

\begin{document}
	\raggedbottom
	\allowdisplaybreaks
	\title{Performance  of UAV-assisted D2D Networks in the Finite Block-length Regime }
	\author{Mehdi~Monemi and Hina Tabassum, \textit{Senior Member}, IEEE
		\thanks{This research was supported by a Discovery Grant funded by the Natural Sciences and Engineering Research Council of Canada. M. Monemi was a Visiting Researcher at the Dept. of Electrical Engineering and Computer Science, York University, Canada and is currently with the Dept. of Electrical and Computer Engineering, Salman Farsi University of Kazerun, Kazerun, Iran (email: monemi@kazerunsfu.ac.ir). H. Tabassum is with the Dept. of Electrical Engineering and Computer Science, York University, Canada.}
	}

	\maketitle
% 	\vspace{-15mm}
	\begin{abstract}
		We develop a comprehensive framework to characterize and optimize the performance of a unmanned aerial vehicle (UAV)-assisted D2D network, where D2D transmissions underlay  cellular transmissions.  Different from conventional non-line-of-sight (NLoS) terrestrial transmissions, aerial transmissions are highly likely to experience line-of-sight (LoS).  As such, characterizing the performance of mixed aerial-terrestrial networks with accurate fading models is critical to precise network performance characterization and resource optimization. We first characterize closed-form expressions for a variety of performance metrics  such as  frame decoding error probability (referred to as reliability), outage probability, and ergodic capacity of users. The terrestrial and aerial transmissions may experience either LoS Rician fading or NLoS Nakagami-m fading with a certain probability. 
		% Then, we develop a novel subchannel and power allocation framework  for  UAV-assisted D2D networks, where D2D transmissions underlay  cellular transmissions.  
		Based on the derived expressions, we  formulate a hierarchical bi-objective mixed-integer-nonlinear-programming (MINLP) problem  to minimize the total transmit power of all users and maximize the aggregate throughput of D2D users subject to quality-of-service (QoS) measures (i.e., reliability and ergodic capacity) of cellular users. We model the proposed problem as a bi-partite one-to-many matching game. To solve this problem, we first obtain the optimal closed-form   power allocations for each D2D and cellular user on any possible subchannel, and then incorporate them to devise efficient subchannel and power allocation algorithms. Complexity analysis of the proposed algorithms is presented. Numerical results verify the accuracy of our derived expressions and  reveal the significance of aerial relays compared to ground relays in increasing the throughput of D2D pairs especially for distant D2D pairs.

	\end{abstract}
	%...................................................................................................................................
	% keywords
	
	%\IEEEpeerreviewmaketitle
	
	%...................................................................................................................................
	% Introduction
	\thispagestyle{empty}
	
	\section{Introduction}
	
	Device-to-device (D2D) transmissions  offer ubiquitous and ultra-reliable connectivity to diverse device types in 5G and beyond 5G (B5G) networks \cite{hoyhtya2018review}. Among the two distinguished modes of D2D transmissions (i.e., overlay and underlay), the underlay mode where D2D transmissions coexist with the legacy cellular transmissions is relatively appealing to network operators due to its efficient spectrum utilization  %\cite{osti2016performance,liu2012mode}.
	% \cite{osti2016performance,kusaladharma2017interference}
	\cite{osti2016performance}.
	Typically, D2D communication facilitates direct communication between physically nearby devices, without the intervention of a base station (BS).   To enable  D2D communications, 3GPP introduced LTE-Direct protocol (also known as  Proximity-based Services (ProSe)) in Release-12 \cite{3gpp20123rd}. Nevertheless, the applicability of ProSe remains limited in scenarios  when the devices are outside the network coverage, far from each other, experience poor channel conditions, or create severe interference to cellular users. To address this limitation, relay-assisted D2D functionality  has been introduced in 3GPP Release-13 \cite{3gpp13}. Subsequently, a large body of the  literature considered terrestrial relays with non-line-of-sight (NLoS) transmissions to instigate D2D communications \cite{sang2011survey,kumar2019survey}.

	Along another note, unmanned-aerial-vehicles (UAV) based communication has emerged as a potential technology  to complement terrestrial networks and will be an integral component of B5G wireless networks \cite{zhang2019survey,li2018uav,zhou2018beam,huo2019multi,khawaja2017uav}. \textcolor{black}{For example, in \cite{zhou2018beam}, the authors have proposed robust beam management and network self-healing mechanisms for millimeter wave (mm-wave) UAV mesh networks. In \cite{huo2019multi}, a mm-wave  distributed phased-arrays architecture and proof-of-concept (PoC) designs for mobile user equipment and UAVs were proposed. {Also, in \cite{khawaja2017uav}, the authors have studied the behaviour of mmWave air-to-ground channels by using two ray propagation model and employing ray tracing simulations for 28GHz and 60GHz frequencies.}} 
	% \textcolor{red}{There exist some works studying aerial LoS channel characterization without considering fading. For example in \cite{khawaja2017uav}, the authors have studied the behaviour of mmWave air-to-ground channels by using two ray propagation model and employing ray tracing simulations for 28GHz and 60GHz frequencies.}
	
	% As such,  UAV-assisted D2D communications can ameliorate the performance of conventional relay-assisted D2D communications with appropriate resource management mechanisms that guarantee quality of service (QoS)  at cellular users (such as reliability, outage probability, latency, and data rate, etc.). 
	%  \cite{global2015road,shah20185g}.   \cite{takahashi2018novel, baek2017optimal,wang2018spectrum}.
	
	Different from the terrestrial infrastructure experiencing mostly NLoS transmissions, UAVs can be deployed flexibly in three dimensions and can offer strong line-of-sight (LoS) connectivity.  %\cite{li2018uav}. 
	Nevertheless, in emerging 5G/B5G mixed aerial and terrestrial networks, we encounter  {\it statistically distinct LoS or NLoS fading channels} and it is crucial to understand their impact on the  network performance metrics and resource allocation schemes. To this end, this paper answers the fundamental questions such as {(i)~how to  characterize key performance metrics such as  outage probability, ergodic capacity, and decoding error probability of users in a {UAV-assisted} D2D underlay network while accurately modeling mixed LoS/NLoS fading channels, (ii)~how to use the derived expressions  to efficiently manage network resources, and  (iii)~in which scenarios UAV-assisted D2D communications can be beneficial. }
	
	The impact of fading channel statistics on performance metrics  can be computed by averaging over the entire distributions of the fading random variables.
	This can be done in two ways, i.e., (i) by computing the optimal solutions  and resource allocations per fading channel realization. Then, solving the instantaneous optimization (resource allocation) problem for a large number of channel realizations and averaging over all channel realizations to compute the  optimal solutions and performance metrics, or (ii) characterizing the statistically averaged performance metrics and then using them to formulate and solve the required optimization (resource allocation) problems. The latter approach enables us (1) to capture the impact of the random fading channel statistics on important performance metrics without solving the resource allocation problem for each channel realization and  does not require averaging over all possible fading channel conditions, and (2)  perfect knowledge of rapidly varying channel state information (CSI) is not needed.

	\subsection{Background Work}
	
	To date, a plethora of research works considered {\it instantaneous optimization of the subchannel and power allocations} of D2D users   assuming either no fading  
	% 	\cite{wang2018spectrum,ma2018full,niu2017device,guizani2016mmwave,liu2017coalition}
	\cite{wang2018spectrum,ma2018full,monemi2016characterization}
	or Rayleigh fading %\cite{chen2019resource,chen2018resource,selim2019outage,sekander2016decoupled}.
	\cite{sekander2016decoupled,chen2018resource}.
	In \cite{wang2018spectrum}, a successive convex optimization based resource allocation scheme was proposed to enhance the performance of D2D transmissions  via UAV relays.  In \cite{ma2018full},  the authors studied relay-assisted D2D communications in mm-wave networks with full-duplex relays that are equipped with directional antennas and obtained a Pareto-optimal scheduling solution.
	%In \cite{niu2017device}, an energy-efficient multicast scheduling scheme was proposed for mm-wave network which utilized multi-hop D2D communications.
	%In \cite{guizani2016mmwave}, the resource allocation for D2D pairs in the  mm-wave spectrum was formulated as NP-hard optimization problem, and then a heuristic algorithm was proposed to allocate the resource blocks (RBs) to underlaid D2D pairs. 
	In \cite{monemi2016characterization}, the authors characterized the feasible region of interference in a D2D underlaid multi-cell cellular network and then devised efficient and reliable resource management schemes for network users. 
	%A coalition game for user association and bandwidth allocation in ultra-dense mm-wave networks was proposed in \cite{liu2017coalition}, wherein the proposed algorithm iteratively converges to a local sub-optimal solution. 
	In \cite{sekander2016decoupled}, a two-sided resource allocation matching game is devised for decoupled uplink-downlink multi-tier full-duplex networks under Rayleigh fading.
	In \cite{chen2018resource},  a coalition game was established to maximize the sum-rate of underlaid D2D users in mm-wave networks. 
	%In \cite{selim2019outage}, for a Rayleigh-faded UAV-assisted D2D network,  an optimization-based power control scheme was proposed.
	% 	In \cite{zarandi2018resource}, a power allocation scheme for  by considering cross-tier interference was devised for inband full-duplex network and the performance was verified under Rayleigh fading channel.
	
	Recently, few {\it instantaneous  optimization frameworks} were developed considering  Nakagami-m fading with shape factor $m$ for NLoS transmission links \cite{kusaladharma2017interference,she2019ultra ,li2019performance} and Rician fading with shape factor $K$ for LoS transmissions  %\cite{swetha2017selective,mishra2019coverage,sboui2017energy,ernest2019power}. 
	\cite{swetha2017selective,mishra2019coverage,sboui2017energy}. For example, in \cite{kusaladharma2017interference}, by considering Nakagami-m fading channel, a user association and power control scheme was proposed for D2D underlaid cellular networks and then outage expressions have been derived. 
	%In \cite{turgut2017uplink}, a simple D2D/cellular mode selection scheme was proposed without considering fading and then the expressions of the outage probability under Nakagami-m fading were calculated to evaluate the performance of the proposed scheduling scheme.
	Assuming that the 
	Shannon channel capacity is not reachable due to practical limitations, under the finite block-length regime, an ultra-reliable resource allocation scheme is devised
	in \cite{she2019ultra}  for UAV-enabled networks wherein the channel is supposed to be  Nakagami-m fading.
	In \cite{swetha2017selective}, an overlay/underlay mode selection and resource  allocation algorithm was presented for mm-wave D2D networks under Rician fading channels. A coverage constrained scheduling scheme for UAV-assisted heterogeneous network (HetNet) was proposed in \cite{mishra2019coverage} under Rician fading channels. In \cite{sboui2017energy}, an energy-efficient power allocation scheme for UAV cognitive radio networks was proposed with Rician fading channels.

	Nevertheless, in the aforementioned research works, the objective function and constraints are defined for each fading channel realization (which is assumed to be perfectly known) and, subsequently, the optimal solutions  and resource allocations are computed per channel realization. 
	% 	{\it Particularly, the impact of the fading channel statistics on optimal resource allocation solutions and  network performance metrics (such as data rate) cannot be directly captured.}  
	The instantaneous optimization needs to be solved for a large number of channel realizations (assuming perfect CSI) to extract  performance as a function of fading statistics. 
	% 	Evidently,  instantaneous optimization is computationally intensive and may not directly assist in extracting insights related to the fading channel statistics and average network performance indicators such as decoding error probability, coverage probability, and ergodic rate. 
	% 	On the other hand, in emerging 5G/B5G mixed aerial and terrestrial networks, we encounter  {\it statistically distinct LoS or NLoS fading channels} and it is becoming crucial to understand their impact on the average network performance metrics and resource allocation schemes. {\it As such, characterizing the average network performance  taking into account the fading channel statistics and then using it to solve resource allocation problems is indispensable.}

	Recently,  a handful of research works exist that considered the outage probability expressions and ergodic rates in the design of resource allocation algorithms for HetNets \cite{della2017minimum, baek2017optimal, wang2018joint, mishra2019spectral,taniya}. For example, in \cite{della2017minimum}, by considering no interference in an overlay D2D network, a minimum power scheduling scheme under Rician fading  was proposed for full-duplex relay-assisted D2D communication. In \cite{baek2017optimal}, a throughput maximization problem with outage probability constraints was considered for UAV relay systems in the presence of Nakagami-m fading channels. In \cite{wang2018joint},  a matching-game based subchannel and power allocation algorithm was proposed for a throughput maximization problem  considering Rayleigh-faded NLoS signal and Rician faded LoS interference. In \cite{mishra2019spectral}, by deriving the outage of users under Nakagami-m and Rayleigh fading channels, deployment cost efficiency   was optimized in mm-wave networks.

	\subsection{Motivation and Contributions}

	We  develop  a  comprehensive  framework  to  characterize  and  optimize  the  performance  of  a  UAV-assisted D2D network considering Nakagami-fading for NLoS and Rician fading for LoS transmission.  Different  from  conventional NLoS  terrestrial  transmissions,  aerial  transmissions are highly likely to experience LoS. As such, characterizing the performance of mixed aerial-terrestrial networks with accurate fading models is critical to precise network performance  optimization. 
	To date, many research works rely on Rayleigh fading models for both the LoS and NLoS channels. While Rayleigh fading models are mathematically tractable, they are not accurate, especially for LoS modeling. Nakagami-m  and Rician fading  are considered as accurate\footnote{The models can capture a variety of channel conditions and empirical measurements. For example,  $m=1$ in Nakagami-m distribution  and $K=0$ in Rician distribution yield Rayleigh fading which is  commonly used  in  terrestrial wireless networks.}  for NLoS and LoS, respectively. 
	Evidently, statistical characterization of the performance  is even more challenging  due to the coexistence of LoS Rician and NLoS Nakagami-m fading channels in the desired links and/or interfering links (especially when the desired signal experiences Rician LoS and the interfering signal experiences Nakagami-m NLoS or vice versa). 
	
	% To the best of our knowledge, the problem of subchannel and power allocation  with  ergodic capacity  and reliability  constraints under the coexistence of Rician LoS and Nakagami-m NLoS channels have not been studied so far. 
	Furthermore, statistical modeling of the average performance metrics (such as decoding error probability, outage probability, and ergodic capacity) and applying them  to devise efficient resource allocation schemes is another challenge. However, as mentioned earlier,  this approach does not require to compute  optimal solutions for each fading channel realizations (i.e., the optimization problem needs to be solved only once) and the instantaneous CSI is not needed.

	To this end, the contributions of this paper are as follows:
	\begin{itemize}
		\item We consider a UAV-assisted D2D underlaid cellular network. Each D2D pair either selects direct or UAV-assisted relay transmission and all D2D transmission channels are shared with cellular users. Both the aerial and terrestrial transmissions experience LoS Rician fading or NLoS Nakagami-m  fading depending on a specific criterion. We characterize the signal-to-interference (SIR) outage probability and ergodic capacity of cellular user transmissions, D2D transmissions, and UAV relay transmissions, considering Nakagami-m fading for NLoS and Rician fading for LoS channels.

		\item  Under the finite block-length quasi-static regime and considering both LoS and NLoS interference scenarios, we derive analytic expressions for the frame decoding error probability of cellular users, which is also referred to as {\it reliability} \cite{polyanskiy2010channel}.  We propose tight approximations to the decoding error probability and demonstrate through numerical results that the approximate expressions are in close agreement to the exact values. {The derived closed-form  expressions for reliability can guarantee that the cellular users can always decode their block of (n bit) data (with a given probability).
		} 
		
		\item Based on the derived expressions, we formulate a hierarchical bi-objective optimization problem as a mixed-integer-nonlinear-program (MINLP)
		%to minimize the total transmit power of all users and maximize the aggregate throughput of D2D users
		{
			to minimize the total transmit power {(the aggregate transmit power of UAVs as well as  D2D and cellular users)} leading to maximum aggregate throughput of D2D users}, subject to QoS measures {(i.e., reliability and ergodic capacity)} of cellular users. 
		%The ergodic capacity constraint can also be viewed as a latency constraint as we will discuss in Section~V. 
		The formulated problem allocates the power of cellular users as well as the power, subchannel and link-type (direct or relayed) for D2D pairs.
		%, in order to find the minimum power corresponding to the maximum aggregate throughput of D2D pairs. 

		\item {We model the proposed problem as a bi-partite one-to-many matching game. In order to solve the problem, first we calculate the {\em optimal closed form power allocation} of each cellular user and D2D pair on each possible subchannel. Then, to assign the {\em subchannel and link-type for D2D pairs}, we first obtain the global optimal solution to the corresponding one-to-one matching game with no UAV relays, and {then we extend the corresponding procedure to achieve the solution for one-to-many matching game scheduling scheme.}}
		
		\item {Complexity analysis of the proposed algorithm is provided. Numerical results reveal that  by choosing optimum height for UAVs, the performance measures of the proposed resource management schemes can be enhanced to a great extent, specially for distant D2D pairs.}
		
	\end{itemize}
	
	The rest of the paper is organized as follows. In Section II, the system model and assumptions are presented. In Section III, we characterize the outage probability and channel capacity for cellular and D2D users. The frame decoding error probability of cellular users is characterized in Section IV. Resource allocation problem and solution are presented in Section V, and finally simulation results and conclusions are provided  in Sections VI and VII, respectively.
	
	\section{System Model and Assumptions}
	\label{sec:system_model}
	
	\begin{figure}
		\centering
		\includegraphics [width=254pt]{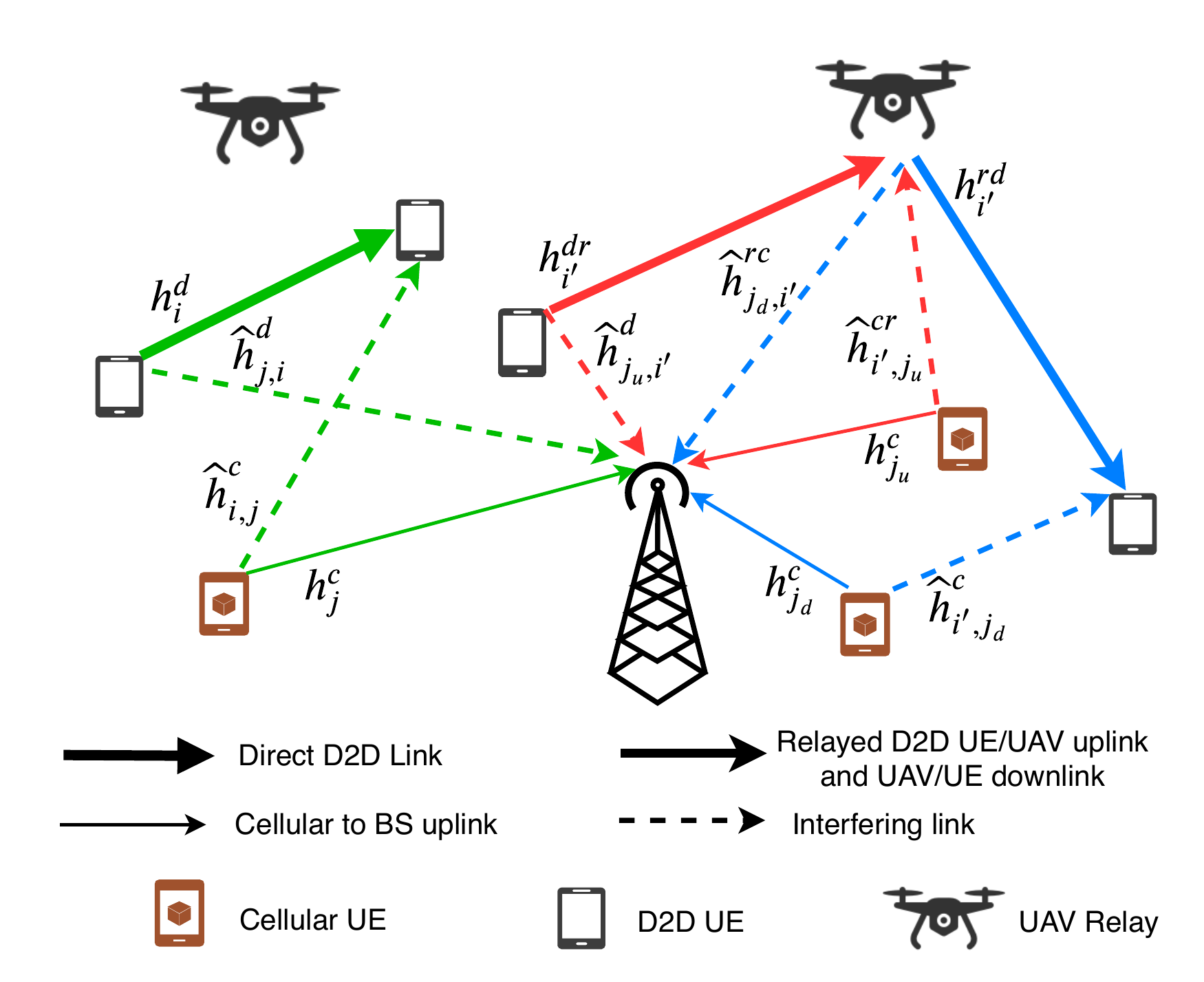} \\
% 		\vspace{-15pt}
		\caption{Graphical illustration of D2D underlaid cellular network consisting of terrestrial BS, several cellular users, direct D2D and relayed D2D pairs. The subchannel allocated to cellular user $j$ is reused by D2D pair $i$ and corresponding links are shown in green color. The relayed D2D pair $i'$  reuses the subchannels allocated to cellular users $j_u$ and $j_d$ in the uplink and downlink paths whose links are depicted in red and blue colors, respectively.
		} 
		\label{fig:structure}
% 		\vspace{-10mm}
	\end{figure}
	\subsection{Network Model}
	Consider a network consisting $M^c$ cellular users and $M^d$ D2D pairs and a set of one or more rotary-wing UAVs\footnote{Rotary-wing UAVs can hover over a certain location  to  ensure  continuous  coverage. 
		Rotor  blades  in the rotary-wing UAV work  exactly in the same  way  as  a  fixed  wing, however  constant  aircraft forward movement is not needed to produce airflow over the blades. 
		% 	Instead, the blades themselves are in constant movement which produces the required airflow over their airfoil to generate lift. 
		As such, a rotary-wing UAV can be considered as predeployed by the network operators for network performance assessment and resource allocation/optimization purposes.}
	with predefined spatial deployment. Let $\M^c$ be the set of cellular users whose minimum QoS requirements must always be guaranteed, and $\M^d$ be the set of D2D pairs who use available cellular resources, as long as minimum QoS requirements of all cellular users are preserved. D2D pairs  with strong channels may choose direct connection, while those with poor direct channel gain (i.e., those D2D pairs whose transmitters and receivers are rather far or potentially NLoS) may be served by the UAV relay. From now on, we will call the former and latter D2D pairs as \textit{direct} D2D and \textit{relayed} D2D pairs, respectively.  It is seen in Fig. \ref{fig:structure} that D2D pair $i\in\M^d$ and $i'\in\M^d$ have established transmission through direct and relayed  links, respectively.
	%The direct D2D pair $i$ and cellular user $j$ are using a shared subchannel and thus impose interference on each other.
	{The relayed transmissions use decode-and-forward (DF) relaying mode.}
	
	\subsection{Channel Model}
	Let $i\in\M^d$ be the direct D2D pair reusing the subchannel of cellular user $j\in\M^c$. Also, let $i'\in\M^d$ be the relayed D2D pair reusing the subchannels of the cellular users $j_u$ and $j_d$ in the uplink and downlink paths, respectively (see Fig.~\ref{fig:structure}). Let $h_i^d$ denotes the (power) channel gain between the transmitter and receiver of direct D2D pair $i\in\M^d$ and $h_j^c$ denotes the channel gain between the cellular user $j$ and ground BS. Also, let $h_{i'}^{dr}$ and $h_{i'}^{rd}$ denote, respectively, the relayed uplink and downlink channel gains for D2D pair $i'$. For all relayed D2D pairs, we consider that the uplink and downlink channels are orthogonal to each other and may be shared with the subchannels of two different cellular users.
	
	The interfering channel gain imposed from cellular user $j$ on the receiver of direct D2D pair $i$ is denoted by $\widehat{h}_{i,j}^c$, whereas that imposed from the transmitter of direct D2D pair  $i$ to the receiver of cellular user $j$ (i.e., BS) is denoted by $\widehat{h}_{j,i}^d$. { Also, for relayed D2D pair $i'$, the interference imposed from  cellular user $j_u$ on the UAV relay is denoted by $\widehat{h}^{cr}_{i',j_u}$ and the interference from the UAV relay to the receiver of cellular user $j_d$ (i.e., BS) is denoted by  $\widehat{h}^{rc}_{j_d,i'}$.} We model any channel gain $h$ (whether desired or interfering link) as follows:
	\begin{align}
	\label{eq:h_fading}
	&h = x \mathbb{E}[h]  \quad \mathrm{where} \\& \mathbb{E}[h] \mathrm{(dB)}=-\begin{cases}
	\mu_L+ 10\beta_L \log_{10} d,\ \ \ \ \textrm{ for LoS link},
	\notag\\
	\mu_N+ 10\beta_N \log_{10} d, \ \ \ \textrm{ for NLoS link},
	\end{cases}
	\end{align}
	where $\mathbb{E}[h]$ is the expected value of $h$ considering the slow fading coefficients (i.e., path-loss and shadowing), and $x$ is unit mean random variable reflecting the fast fading coefficient which might not be exactly estimated.	Also, $d$ is the distance of the link, $\mu_L$ and $\mu_N$ are frequency dependent components of the path loss, and $\beta_L$ and $\beta_N$ are the path loss exponents for the LoS and NLoS cases, respectively. We model the LoS and NLoS fading channels with Rician \footnote{The approximation of Rician fading with Nakagami-m fading   is  generally  not  accurate for LoS channels \cite{annamalai2001simple},  and  thus in this work  we  have  employed accurate Rician fading model for LoS channels.} distribution  having shape factor $K$ and Nakagami-m distribution   with shape factor $m$, respectively, as follows:
	\begin{align}
	\label{eq:Rician_pdf}
	%f_r(x,K)=e^{-K}e^{-x} I_0(\sqrt{4Kx}),
	f_R(x,K)&=(K+1)e^{-K-(K+1)x} I_0(\sqrt{4K(K+1)x}), \\
	\label{eq:Nakagami_pdf}
	f_G(x,m)&=\frac{m^m x^{m-1}}{\Gamma(m)} e^{-mx},
	\end{align} {where $I_0(.)$ and $\Gamma(.)$ are modified Bessel and gamma  functions respectively}.
	The fading models and notations considered for different aerial and terrestrial  links are listed in {\bf Table~1}. The LoS probability for aerial links is given by \cite{al2014optimal,mozaffari2016unmanned}:
	\begin{align}
	\label{eq:p_LoS_per_teta}
	p_L(\theta)=\frac {1}{1+C \mathrm{exp} (-B [ \theta -C ])},
	\end{align}
	where $\theta$ is the angle between UAV and ground node, and $B$ and $C$ are constants related to environment. 
	%        For transmissions between ground devices (i.e., transmission between each D2D transmitter and receiver and interfering links from a cellular user to D2D receivers), we assume that the  LoS probability depends on the distance between the transmitter and receiver, i.e.,
	%        
	%        \begin{align}
	%        \label{eq:p_LoS_per_d}
	%        p_L(d)=
	%            \begin{cases}
	%             d/d^{max}, & \textrm{If $d < d^{max}$ },
	%            \\
	%            1, & \textrm{If $d \geq d^{max}$ }.
	%            \end{cases}
	%        \end{align}
	For transmissions between ground devices (i.e., transmission between each D2D transmitter and receiver, and interfering links from a cellular user to D2D receivers), we consider the 3GPP low altitude LoS probability model for urban environments as
	\begin{align}
	\label{eq:p_LoS_per_d}
	p_L(d)=	\min \left(\frac{d_1}{d}, 1\right)\left(1-\exp\left(-\frac{d}{d_2}\right)\right)  + \exp\left(-\frac{d}{d_2}\right),
	\end{align}
	where $d_1=18$m and $d_2=63$m  \cite{access2016further}.
	
	\begin{table*}
% 		\vspace{-10mm}
		\centering
		\caption{Fading models and notations considered for different transmission channels}
		\label{tbl:fadings}
% 		\vspace{-5mm}
		\begin{tabular}{|c|c|c|c|c|}
			\hline
			\multicolumn{2}{|c|}{}                                                                                    & \textbf{
				\begin{tabular}[c]{@{}c@{}}
					Desired Channel \\ Fading and Notation\end{tabular}}

			& \textbf{\begin{tabular}[c]{@{}c@{}}
					Interference Channel  \\ Fading and Notation\end{tabular}} \\ \hline
			\multicolumn{2}{|c|}{\textbf{$j$-th Cellular User}}                                                              & Nakagami-m, $h^c_j$                                                           & Nakagami-m/Rician , $\widehat{h}_{i,j}^c$                                                        \\ \hline
			\multicolumn{2}{|c|}{\textbf{$i$-th Direct D2D Pair}}                                                            & Nakagami-m/Rician, $h_i^d$                                                    & Nakagami-m, $\widehat{h}_{j,i}^d$                                                         \\ \hline
			\multirow{2}{*}{\textbf{\begin{tabular}[c]{@{}c@{}}UAV Relayed \\$i^\prime$-th D2D Pair\end{tabular}}} & \textbf{Uplink}   & Nakagami-m/Rician, $h_{i'}^{dr}$                                                               & Nakagami-m/Rician, $\widehat{h}_{i',j_u}^{cr}$                                                                    \\ \cline{2-4} 
			& \textbf{Downlink},  & Nakagami-m/Rician,  $h_{i'}^{rd}$                                                               & Nakagami-m/Rician, $\widehat{h}_{j_d,i'}^{rc}$                                                     \\ \hline
		\end{tabular}
		%\vspace{-10mm}
	\end{table*}

	\subsection{SINR Model for Direct Links}
	\textcolor{black} {Consider that direct D2D pair $i\in\M^d$ and cellular user $j\in\M^c$ are allocated same subchannel. The SINR at the D2D receiver is expressed as follows:
		%	\begin{align}
		%	    \label{eq:gamma}
		%		%\label{eq:gamma_d2d_direct}
		%		\gamma_i^d =
		%			\dfrac{h^d_i p^d_i}{\mathcal{I}_i^d + \widehat{h}^c_{i,j} p^c_j} ,
		%        \quad \quad
		%		%\label{eq:gamma_cellular}
		%		\gamma_j^c =
		%		\dfrac{h^c_j p^c_j}{\mathcal{I}_j^c +  \widehat{h}^d_{j,i} p^d_i },
		%	\end{align}
		\begin{align}
		\label{eq:gamma}
		%\label{eq:gamma_d2d_direct}
		\gamma_i^d =
		\dfrac{h^d_i p^d_i}{\Iid + \widehat{h}^c_{i,j} p^c_j} ,
		\end{align}
		where $p_i^d$ and $p_j^c$ are the transmit powers of  D2D pair $i$ and cellular user $j$ respectively,  $\widehat{h}^c_{i,j} p^c_j$ is the dominant co-channel intracell interference imposed from cellular user $j$ on D2D pair $i$, and
		$\Iid=\Izeroid +\npower$, in which $\npower$ is the noise power, and $\Izeroid$ is the interference imposed on D2D pair $i$ from all sources other than cellular user $j$ (which mainly results from the co-channel intercell interference due to frequency reuse in non-adjacent cells/sectors). Similarly, the SINR of cellular user $j$ is expressed as
		\begin{align}
		\label{eq:gamma_j}
		%\label{eq:gamma_cellular}
		\gamma_j^c =
		\dfrac{h^c_j p^c_j}{\Ijc +  \widehat{h}^d_{j,i} p^d_i },
		\end{align}
		where $\widehat{h}^d_{j,i} p^d_i$ is the dominant co-channel intracell interference imposed from D2D pair $i$ on cellular user $j$, and $\Ijc=\Izerojc +\npower$, in which $\Izerojc$ is the interference imposed on cellular user $j$ from all sources other than D2D pair $i$.
		%\footnote{This work considers a single \textit{controlled} interference which enables interference management through a joint power control strategy explored in Section~V. On the other hand, the effect of multiple (uncontrolled) interferers can be  incorporated in our work in a straight-forward manner, i.e.,  by taking the noise floor in (6) and (7) as $\sigma^2=n_0+I_0$, where $n_0$ is the thermal noise at the receiver and $I_0$ is the summation of all interference powers at the corresponding receiver.}
	} Note that the cellular user transmission can be interfered by either (i) the D2D transmitter in direct D2D relaying or (ii) either D2D transmitter or UAV relay in the relaying D2D link, and thus, as seen in Fig. \ref{fig:structure} there is only interfering link for each subchannel at a time. Similarly, the D2D receiver or UAV relay will receive interference from a given cellular user transmission. Let $k$ be some (D2D or cellular) user and $k'$ be its corresponding interfering user (if exists any). This way, if $k\in\M^d$, then we have $k'\in\M^c$, and if $k\in \M^c$, then we have $k'\in \M^d$. Now, in order to make notations simpler, we combine the SINRs in \eqref{eq:gamma} (and also any other possible combinations for the relayed D2D pairs) as follows:
	\begin{align}
	\label{eq:gamma_k}
	\gamma_k =
	\dfrac{h_k p_k}{\textcolor{black}{\Ik} +   \widehat{h}_{k,k'} p_{k'}} ,
	\end{align}
	where $\gamma_k$ can be the SINR of any of the transmissions with D2D pair or cellular user, $h_k$ is the main channel gain of user $k$ and $\widehat{h}_{k,k'}$ corresponds to the interfering channel gain imposed from user $k'$ to user $k$, and finally $p_k$ and $p_{k'}$ are transmit power levels of users $k$ and $k'$ respectively. For example, if $k=i\in \M^d$ is a direct D2D pair interfering with cellular user $j$, we have $h_k$=$h^d_i$, and $\widehat{h}_{k,k'}=\widehat{h}^c_{i,j}$, and if $k=j\in \M^c$ is a cellular user %interfering with  direct D2D pair $i$
	, we have  $h_k$=$h^c_j$, and $\widehat{h}_{k,k'}=\widehat{h}^d_{j,i}$.
	
	\subsection{SINR Model for Relayed Links}
	The relayed transmissions use DF relaying mode. If D2D pair $i'\in\M^d$ establishes a relayed link, the  end-to-end SINR is the minimum of the SINR of the uplink and downlink \cite{riihonen2011hybrid}, i.e.,
	\begin{align}
	\label{eq:relay_SINR1}
	\gamma_{i'}^{d,rel} &=
	\textrm{min}  \left\{ \frac{h_{i'}^{dr} p_{i'}^{dr}} {\textcolor{black}{\mathcal{I}^d_{i',j_u}} + \widehat{h}_{i',j_u}^{cr} p_{j_u}^c } , \frac{h_{i'}^{rd} p_{i'}^{rd}}{\textcolor{black}{\mathcal{I}^d_{i',j_d}} + \widehat{h}_{i',j_d}^c 	p_{j_d}^c  } \right\}.
	% 		\\
	% 		R_{i',s}^{d,rel} &= \mathbb{E}[ \textrm{log} (1+\gamma_{i'}^{d,rel})].
	\end{align}
	where $j_u$ and $j_d$ are the cellular users whose subchannels are reused by relayed D2D pair $i'$ in the uplink and downlink paths respectively, $p_{i'}^{dr}$ is the transmit power of D2D transmitter to the relay, and $p_{i'}^{rd}$ is the transmit power of relay to the receiver of D2D pair $i'$. For brevity, we express all SINRs in the form of \eqref{eq:gamma_k}, and for some relayed D2D pair $k$, we rewrite \eqref{eq:relay_SINR1} as $\gamma_{k}^{d,rel} =
	\textrm{min}  \left\{ \gamma_{k_u}, \gamma_{k_d} \right\},
	$	where $\gamma_{k_u}$ and $\gamma_{k_d}$ are the uplink and downlink SINRs which can be expressed from \eqref{eq:gamma_k} as $
	%\begin{align}
	%\label{eq:gamma_ku_kd}
	\gamma_{k_u} =
	\frac{h_{k_u} p_{k_u}}{\textcolor{black}{ \mathcal{I}_{k_u,k'_u}}  +   \widehat{h}_{k_u, k'_u} p_{k'_u}} ,
	\gamma_{k_d} =
	\frac{h_{k_d} p_{k_d}}{\textcolor{black}{ \mathcal{I}_{k_d,k'_d}} +   \widehat{h}_{k_d, k'_d} p_{k'_d}}.
	%\end{align}
	$
	{For brevity, in the following sections, the notation $i$ may stand for either direct or relayed D2D pair.}
	
	%Comparing \eqref{eq:gamma_ku_kd} and \eqref{eq:relay_SINR1} illustrates the notation matching for each of the notation items.

	% \subsection{Power Consumption}
	% The power consumption due to hovering of UAV \cite{zeng2018energy} is given as:
	% 	\begin{align}
	% 	%\begin{split}
	% 	E^{u}= (P_{h} +p_u)T_u,
	% 	%\end{split}
	% 	\label{eq:EnergyHOver1} 
	% 	\end{align}
	% 	where  $p_{u}$ is the uplink power consumption of BS$_s$ and $P_{h}$ is the propulsion  power consumption of UAV in the hovering state given as follows:
	% 	\begin{equation}\label{Ph}
	% 	 P_h=P_0+P_i=\frac{\delta}{8}\rho s A \xi^3r^3+(1+\kappa)\sqrt{\frac{(mg)^{3}}{{2\rho A}}},  
	% 	\end{equation}
	% 	  where  $\rho$,  $A$, $\xi$, $r$,   $s$, $\delta$, and   $\kappa$ denote the air density (in kg/m$^3$), rotor disc area (in m$^2$), blade angular velocity (in rad/sec), rotor radius (in m),  rotor solidity,  profile drag coefficient, and incremental correction  factor of induced power,  respectively.
	% 	%$T_u$ is the time required for $BS_s$ to transfer the complete data to UAV as $d_{U}=H$ which is same for hovering and communication because UAV hovers as long as all the data is communicated to UAV. Hence,  
	% 	%\begin{align}
	% 	%E_{\mathrm{T}}^{U}=&P_c.T_{U}.
	% %	\label{eq:EnergyComm1}
	% 	%\end{align}
	% The total weight $W$ of the UAV is approximately equal to the gravitational force, i.e. $W = m g$ (in Newton), where $m$ is the UAV mass (in kg) and $g$ denotes earth gravity (in m/s$^2$). 

	\section{Expressions of Outage Probability and Channel Capacity}
	
	We  characterize the outage probability for cellular transmissions, direct D2D, and relayed D2D transmissions considering interference limited regime. The expressions will be used to derive the ergodic capacity  and  the average frame decoding error probability  of the cellular users.

	\subsection{Outage Probability Analysis for a Typical Link}
	Before deriving the outage probability, {we ensure that the interference-limited assumption is always valid, i.e., we ensure that the interference in (\ref{eq:gamma_k}) is much higher than noise power, which leads to the inequality $\underline{p}_{k',k} < p_{k'}$} where
	\begin{align}
	\label{eq:interference_limited_dir_k}
	\underline{p}_{k',k}= {\textcolor{black}{\mathcal{I}_{k,k'}} \widetilde{K} }/ {\mathbb{E}\{\widehat{h}_{k,k'}\}},
	\end{align}
	in which $\widetilde{K} \gg 1$ is a constant. Therefore, the following  outage probability derivations are valid when this inequality holds. For example, if direct D2D pair $i\in\M^d$ and  cellular user $j\in \M^c$ share same subchannel,  \eqref{eq:gamma} and \eqref{eq:interference_limited_dir_k} correspond to $p^d_i>\underline{p}^d_{i,j}$ and $p_j^c>\underline{p}^c_{j,i}$ where
	%\begin{align}
	%	\label{eq:interference_limited_dir_ij}
	$\underline{p}^d_{i,j}= {\npower \widetilde{K} }/ {\mathbb{E}\{\widehat{h}^d_{j,i}\}}$ and $
	\underline{p}^c_{j,i}= {\npower \widetilde{K} }/ {\mathbb{E}\{\widehat{h}^c_{i,j}\}}$.
	%\end{align}
	Similarly, for relayed D2D pairs, if $i\in\M^d$ reuses the subchannel of $j_u\in\M^c$ for the uplink and that of $j_d\in\M^c$ for the downlink, we must have $p_i^{dr}>\underline{p}_{i,j_u}^{dr}$, $p_{j_u}^c>\underline{p}_{j_u,i}^{cr}$, $p_i^{rd}>\underline{p}_{i,j_d}^{rd}$, and $p_{j_d}^c>\underline{p}_{j_d,i}^c$ where
	%\begin{subequations}
	%\label{eq:interference_limited_relay}
	%\begin{align}
	$\underline{p}_{i,j_u}^{dr}= {\textcolor{black}{\mathcal{I}_{j_u,i}} \widetilde{K} }/ {\mathbb{E}\{\widehat{h}^d_{j_u,i}\}} $, $
	\underline{p}_{j_u,i}^{cr}= {\textcolor{black}{\mathcal{I}_{i,j_u}} \widetilde{K} }/ {\mathbb{E}\{\widehat{h}^{cr}_{i,j_u}\}}$, $
	\underline{p}_{i,j_d}^{rd}= {\textcolor{black}{\mathcal{I}_{j_d,i}} \widetilde{K} }/ {\mathbb{E}\{\widehat{h}_{j_d,i}^{rc}\}},$ and $
	\underline{p}^c_{j_d,i}= {\textcolor{black}{\mathcal{I}_{i,j_d}} \widetilde{K} }/ {\mathbb{E}\{\widehat{h}_{i,j_d}^c\}}$. 
	%\end{align}
	%\end{subequations}
	Let $y$ and $x$ denote respectively the unit mean fading variables of the main  and interfering links $k$ and $k'$. The \textit{outage probability function} denoted by $O(\alpha)$ is expressed as 
	\small
	\begin{subequations}
		\label{eq:o_dir1}
		\begin{numcases}{
			%O(\zeta_{k,k'} \gammatarget_k)=\mathrm{Pr}\{\gamma_k < \gammatarget_k \} = \mathrm{Pr}\{y < \zeta_{k,k'} \gammatarget_k
			O(\alpha) = \mathrm{Pr}\{y < \alpha  x \}= 
		}
		\label{eq:o_dir1_LoS_LoS}
		\!\!
		O_{L,L}(\alpha), \ \  \textrm{if $k$ LoS, $k'$ LoS,  }
		\\ \!\!
		\label{eq:o_dir1_LoS_NLoS}
		O_{L,N}(\alpha),\   \textrm{if $k$ LoS, $k'$ NLoS,\ \ \ }
		\\ \!\!
		\label{eq:o_dir1_NLoS_LoS}
		O_{N,L}(\alpha),\  \textrm{if $k$ NLoS, $k'$ LoS, \ \ \ }
		\\ \!\!
		\label{eq:o_dir1_NLoS_NLoS}
		O_{N,N}(\alpha), \ \textrm{if $k$ NLoS, $k'$ NLoS, \ \ \ \ \ \ }
		\end{numcases}
	\end{subequations}  
	\normalsize
	in which the outage representations of $O_{L,L}(\alpha)$, $O_{L,N}(\alpha)$ , $O_{N,L}(\alpha)$, and $O_{N,N}(\alpha)$ are studied later in this section. The outage probability value of user $k$ is then obtained as
	\begin{align}
	\mathrm{Pr}\{\gamma_k < \gammatarget_k \} = \mathrm{Pr}\{y < \zeta_{k,k'} \gammatarget_k   x \} 
	=
	%\stackrel{\triangle}{=}
	O(\alpha)\big|_{\alpha=\zeta_{k,k'} \gammatarget_k}
	\end{align}
	where $\gammatarget_k$ is target-SINR of user $k$, and
	\begin{align}
	\label{eq:alpha_k_dir1}
	\zeta_{k,k'}=  p_{k'}  \mathbb{E}[\widehat{h}_{k,k'}]   /  {p_k \mathbb{E}[h_{k}]}.
	\end{align}
	If the link relating to the main user $k$ and that relating to the interfering user $k'$ are both LoS with Rician shape factors $K$ and $K'$, the outage probability expressed in \eqref{eq:o_dir1_LoS_LoS} can be obtained as follows
	%applied_for_single at the bottom of the page 
	\cite{khoolenjani2009ratio}:
	\begin{figure*}[!b]
	\hrulefill
	%\begingroup
	\begin{align}
	\label{eq:o_dir1_LoS_LoS_output}
	O_{L,L}(\alpha)= 1-\left\{  \frac{ (K'+1)e^{-K-K'}   }{\alpha (K+1) + (K'+1) }
	\sum_{m=0}^{\infty} \frac{K^m}{m!} \left[ \sum_{n=0}^{\infty} \frac{1}{n!} \left( \frac{\sqrt{K'(K'+1)}}{\alpha(K+1)  + (K'+1)}  \right)^{2n}
    \!\!\!\!
	\:_2{F_1}(-n,-n,m+1;\frac{ \alpha K}{2K'(K'+1)} \right]
	\right\},
	\end{align}
	%\endgroup
	\end{figure*}
	where $_2{F_1}(\cdot)$ is the Gauss's hypergeometric function.
	To obtain  closed-form expressions for (\ref{eq:o_dir1}b-\ref{eq:o_dir1}d), we consider the shape factor of Nakagami-m fading to be integer.
	\newcommand{\KP}{(K+1)}
	\begin{lemma}
		\label{lm:outage_L_N}
		For the case when the links of users $k$ and $k'$ are LoS and NLoS with shape factors $K$ and $m$ respectively, the outage  probability expression in \eqref{eq:o_dir1_LoS_NLoS} is obtained as
		\begin{multline}
		\label{eq:outage_L_N}
		O_{L,N}({\alpha})= 1- \frac{{m}^{m} e^{-K}}{\Gamma({m}) \left( {m}+\alpha (K+1) \right)^{{m}}} \times 
		%applied_for_single 
		\\
		\frac{d^{{m}-1}}{d \theta^{{m}-1}} \left[ \frac{\theta^{{m}-1}}{1-\theta} ( e^K - \theta e^{K\theta})\right]_{\theta=\frac{\alpha (K+1)}{{m}+\alpha (K+1)}}.
		\end{multline}
		\begin{proof}
			See \textbf{Appendix A}%\ref{apx:outage_L_N}}.
		\end{proof}
	\end{lemma}
	Once $O_{L,N}(\alpha)$ is obtained through \eqref{eq:outage_L_N}, $O_{N,L}(\alpha)$ can be obtained from the following.
	\begin{lemma}
		For the case when the links of users $k$ and $k'$ are NLoS and LoS with shape factors $m$ and $K$ respectively, the outage probability expression in \eqref{eq:o_dir1_NLoS_LoS} is obtained as follows:
		\begin{align}
		\label{eq:outage_N_L}
		O_{N,L}(\alpha) =1 - O_{L,N}(1/\alpha).
		\end{align}
	\end{lemma}
	Finally, we derive the outage probability of $O_{N,N}(\alpha)$ as follows.
	\begin{lemma}
		\label{lm:outage_N_N}
		Let the links of users $k$ and $k'$ both be NLoS with Nakagami-m shape factors $m$ and $m'$ respectively. The outage probability expression in \eqref{eq:o_dir1_NLoS_NLoS} is obtained as
		\begin{align}
		\label{eq:outage_N_N}
		O_{N,N}(\alpha)=
		1 - \frac{m'^{m'}}{\Gamma(m')} \sum_{k=0}^{m-1} \frac{(m \alpha)^k \Gamma(k+m')}{k! (m' +m \alpha)^{k+m'}}.
		\end{align}
	\end{lemma}
	\newcommand{\KK}{\frac{m^{m} (K+1) }{\Gamma(m)e^{K}}}
	\newcommand{\KKKK}{\frac{(K+1)}{e^{K}}}
	\begin{proof}
		See \textbf{Appendix B.}%\ref{apx:outage_N_N}}.
	\end{proof}
	\subsubsection{Outage Probability for Cellular Users}
	As discussed before, the desired links between the cellular users and their serving BS and interfering links from D2D transmitters to the cellular BS follow Nakagami-m fading channels. However, for relayed D2D links, the interfering link from UAV to cellular BS may follow LoS Rician fading, therefore in general, the outage probability of a cellular user $j$ is obtained as $O^c(\alpha)\big|_{\alpha=\zeta_{j,i} \gammatarget_j}$ where
	\begin{subequations}
		\label{eq:O_cell}
		\begin{numcases}{O^c(\alpha)=}
		\label{eq:O_c_L}
		O_{N,L}(\alpha),   &\textrm{for LoS interference,}
		\\
		\label{eq:O_c_N}
		O_{N,N}(\alpha),  &\textrm{for NLoS interference.}
		\end{numcases}
	\end{subequations}
	
	\subsubsection{Outage Probability for Direct D2D Pairs}
	Suppose some direct D2D pair $i$ reuses the subchannel of cellular user $j$.
	Depending on the distance between D2D transmitter and receiver and also the distance between the D2D receiver and the corresponding interfering cellular user, there exists positive probability values for LoS and NLoS links for each of the desired and interfering links. Therefore, the outage in \eqref{eq:o_dir1} may have any of the four possible expressions, and thus the outage value is obtained by $O^{d,dir}(\alpha)\big|_{\alpha=\zeta_{i,j}\widehat{\gamma}_i}$, where $O^{d,dir}(\alpha)$ is obtained by 
	\begin{subequations}
		\label{eq:o_d_dir1}
		\begin{numcases}{O^{d,dir}(\alpha)= }
		O_{L,L}(\alpha), &\textrm{if $i$ LoS, $j$ LoS,}
		\\ 
		\label{eq:o_D_dir1_LoS_NLoS}
		O_{L,N}(\alpha), &\textrm{if $i$ LoS, $j$ NLoS,}
		\\ 
		\label{eq:o_d_dir1_NLoS_LoS}
		O_{N,L}(\alpha),  &\textrm{if $i$ NLoS, $j$ LoS,}
		\\ 
		\label{eq:o_d_dir1_NLoS_NLoS}
		O_{N,N}(\alpha),   &\textrm{if $i$ NLoS, $j$ NLoS.}
		\end{numcases}
	\end{subequations}
	
	\subsubsection{Outage Probability for Relayed D2D Pairs}
	%Now consider the relayed link where a UAV receives the data from a D2D transmitter on a shared subchannel and relays the data through some other shared subchannel to the corresponding D2D receiver. 
	Now consider a relay-assisted D2D link. Depending on the the relative positions of the UAV, cellular BS, and the transmitter and receiver of the D2D pair,  the desired and interfering signals in both the uplink and downlink  may follow LoS or NLoS fading. The following lemma can be easily verified. 
	\begin{lemma}
		Consider the relayed D2D pair $i\in\M^d$ where UAV receives data from the D2D transmitter on the uplink subchannel shared with cellular user $j_u\in\M^c$ and relays data to the corresponding D2D receiver on the downlink subchannel shared with cellular user $j_d\in\M^c$. The outage probability can thus be  obtained as $O^{d,rel}(\alpha_u,\alpha_d)\big|_{\alpha_u=\zeta_{i,j_u} \gammatarget_i,\alpha_d=\zeta_{i,j_d}\gammatarget_i}$ where
		\begin{align}
		\label{eq:o_rel}
		O^{d,rel}(\alpha_u,\alpha_d)=
		1-\left(1-O(\alpha_u)\right)\left(1-O(\alpha_d)\right),
		\end{align}
		in which $O(\alpha_u)$ and $O(\alpha_d)$ correspond to the outage probability functions of the uplink and downlink, respectively, obtained from \eqref{eq:o_dir1}. Also, $\zeta_{i,j_u}$ and $\zeta_{i,j_d}$ are obtained from 
		\eqref{eq:alpha_k_dir1} as 
		$
		\zeta_{i,j_u}={p_{j_u} \mathbb{E}[\widehat{h}_{i,j_u}^{cr}]}   /  {p_i^{dr} \mathbb{E}[h_i^{dr}]},
		\zeta_{i,j_d}={p_{j_d} \mathbb{E}[\widehat{h}_{i,j_d}^c]}   /  {p_i^{rd} \mathbb{E}[h_i^{rd}]}.
		$
		Depending on the fading in the desired and interfering link for uplink and downlink,  $O(\zeta_{i,j_u}\gammatarget_i)$ and $O(\zeta_{i,j_d}\gammatarget_i)$ in \eqref{eq:o_rel} can follow from any one of \eqref{eq:o_dir1_LoS_LoS}, \eqref{eq:o_dir1_LoS_NLoS}, \eqref{eq:o_dir1_NLoS_LoS}, or \eqref{eq:o_dir1_NLoS_NLoS}.
	\end{lemma}
	
	\subsection{Expressions of Link Capacity for D2D and cellular users}
	Based on the expressions obtained for the outage probability of different links, we can now calculate the channel capacity for different users. For a given user having the SINR $\gamma$, the (ergodic) channel capacity  is obtained according to the Shannon Theorem as follows:
	\begin{multline}
	\label{eq:R_calc}
	\mathbb{E}[ \log_2 (1+\gamma)]  
	= \int_{0}^{\infty} \log_2 (1+\gamma) f_{\boldsymbol{\gamma}}(\gamma)  d\gamma 
	%\lim_{\gamma \rightarrow \infty} \left\{ \left[\log_2(1+\gamma)F_{\boldsymbol{\gamma}}(\gamma)\right]_0^\gamma - \frac{1}{\ln{2}} \int_{0}^{\gamma} \frac{F_{\boldsymbol{\gamma}}(\gamma)}{1+\gamma} d\gamma \right\} \nonumber  
	\\
	 	 = 
	\frac{1}{\ln 2} \int_{0}^{\infty} \frac{1- F_{\boldsymbol{\gamma}}(\gamma)}{1+\gamma} d\gamma,
	\end{multline}
	where $f_{\boldsymbol{\gamma}}(\gamma)$ and $F_{\boldsymbol{\gamma}}(\gamma)$ are the probability density function (pdf) and cumulative density function (cdf) of the SINR, respectively.
	Therefore, the channel capacity of cellular users and D2D pairs are obtained as follows:
	\begin{subequations}
		\label{eq:Rate_Statements}
		\begin{align}
		% 			\label{eq:R_dir0}
		% 			&R^{dir_0}(\zeta_{k}^{dir_0})=f_r\left( O^{dir_0}(\zeta_{k}^{dir_0} \gamma) \right),		\\
		\label{eq:R_c_dir}
		&R^c(\zeta_{j,i})=f_r\left( O^c(\zeta_{j,i} \gamma) \right),
		\\
		\label{eq:R_d_dir}
		&R^{d,dir}(\zeta_{i,j})=f_r\left( O^{d,dir}(\zeta_{i,j} \gamma) \right),
		\\
		\label{eq:R_d_rel}
		&R^{d,rel}(\zeta_{i,j_u},\zeta_{i,j_d}) =  f_r\left( O^{d,rel}(\zeta_{i,j_u} \gamma,\zeta_{i,j_d} \gamma) \!\right),
		\end{align}
	\end{subequations}
	where $f_r(g(\gamma))=\frac{1}{\mathrm{ln}(2)}\int_{0}^{\infty} \frac{1-g({ \gamma})}{1+\gamma} d\gamma$, \eqref{eq:R_c_dir} obtains the channel capacity of cellular user $j$ sharing subchannel with D2D pair $i$, \eqref{eq:R_d_dir} obtains the channel capacity of direct D2D pair $i$ reusing the subchannel of cellular user $j$, and \eqref{eq:R_d_rel} obtains the channel capacity of relayed D2D pair $i$ reusing the subchannels of cellular users $j_u$ and $j_d$ in the uplink and downlink paths, respectively.

	\section {Expressions of Frame Decoding Error Probability}
	Channel capacity is the largest achievable rate at which the information can be transmitted regardless of the decoding error probability, however, due to hardware limitations such as the modulation and demodulation techniques employed in the transmitter and receiver, there always exists a gap between the Shannon capacity and the achievable rate. Consider $n$ to to be the blocklength and $\eps_n$ to be the corresponding decoding error probability. The maximum achievable bit-rate per symbol for quasi-static channels under finite block-length ($n<100$) is tightly approximated by the following equation \cite{yang2014quasi,she2019ultra}.
	\begin{align}
	\label{eq:r_frame}
	R^{*}(n,\eps_n) \approx C -\sqrt{\frac{V}{n}} Q^{-1} (\eps_n), 
	\end{align}
	where $C=\log_2 (1+\gamma)$ is the channel capacity (in bits/second/Hz),  $ V=1-1/\left(1+\gamma \right)^2$ is called the channel dispersion, $Q^{-1}$ is the inverse of Marcum Q-function $(Q(x)=\frac{1}{\sqrt{2 \pi}} \int_x^{\infty} e^{-(\frac{u^2}{2})}du)$.
	Assume that the transceiver modulation-demodulation mechanisms limit the achievable data rate to some coefficient of the channel capacity, i.e.,  $R^*(n,\eps)=\xi \log_2(1+\gamma)$ where $\xi<1$ is  constant. Then from \eqref{eq:r_frame} we have
	$C(1-\xi)=\sqrt{\frac{V}{n}} Q^{-1} (\eps_n)
	$, and by inversing this equation, we have%write $\eps_n$ as follows:
	\begin{align}
	\label{eq:Qu}
	&\eps_n (\gamma)
	 = Q\left((1-\xi)\sqrt{\frac{n}{V}}C \right)
\notag\\&	=Q\left(\!(1-\xi)\sqrt{\frac{n}{1-\frac{1}{(1+\gamma)^2}}} \log_2 (1+\gamma)\! \right)
	\!\!\!= Q\left(u(\gamma) \right), 
	\end{align}
	where $u(\gamma)=\frac{(1-\xi)\sqrt{n}}{\ln 2} \left(1-\frac{1}{(1+\gamma)^2}\right)^{-0.5} \ln (1+\gamma)$.
	\begin{theorem}
		\label{th:decoding}
		For cellular user $j$ sharing a subchannel with D2D pair $i$ and utilizing frames of block-length $n$, the frame decoding error probability under quasi-static regime % denoted by $\overline{\eps}_n
		is obtained as
		\begin{multline}
		\label{eq:eps_n_average}
		\mathbb{E}\{\eps_n(\gamma)\} \approx 
		\sum_{i=1}^{L-1}(\omega_{i-1}-\omega_i)H(\zeta_{j,i},\gamma_i) + \omega_L H(\zeta_{j,i},\gamma_L)\\
		\stackrel{\Delta}{=}
		\overline{\eps}_n(\zeta_{j,i}),
		% +
		%     \\
		%     \left(\omega_2 \gamma_1 - \omega_2 \gamma_2 \right) O^{cell}(\alpha \gamma_1)
		\end{multline}
		where $L>1$ is the integer approximation factor\footnote{We will show through numerical results that assigning $L=4$ which is equivalent to a 4-level piece-wise linearization of \eqref{eq:Qu}, tightly fits the corresponding exact expression with negligible error.}, $\omega_i=\frac{0.5/L}{\gamma_i - \gamma_{i-1}}$, $\gamma_i=\eps_n^{-1}(0.5(1-i/L))$ for $0\leq i \leq L-1$, and $\gamma_L=\eps_n^{-1}(\Delta)$, in which $\Delta \ll 0.5/L$ is an arbitrary small value constant, and
		\begin{subequations}
			\begin{numcases}{\!\!\!H(\alpha,\gamma)=}
			\label{eq:H_N}
			H_N(\alpha,\gamma), &\!\!\!\!\!\!\!\textrm{for NLoS interference,}
			\\
			\label{eq:H_L}
			H_L(\alpha,\gamma), &\!\!\!\!\!\!\!\textrm{for LoS interference,}
			\end{numcases}
		\end{subequations}
		where
		\begin{align}
		&H_N(\alpha,\gamma) \!=\! \gamma - 
		\frac{m'^{m'}}{\Gamma(m')} \sum_{k=0}^{m-1} 
		\frac{\Gamma(k+m')(m\alpha)^k \gamma^{k+1}}{(k+1)! (m')^{k+m'}} 
		\notag\\
		&\times {}_2 F_1 \left(k+1,k+m',k+2,-\frac{m\alpha\gamma}{m'}\right), 
		\\
		& H_L(\alpha,\gamma)=\frac{ e^{-K} (m\alpha )^m  \gamma^{m+1}} {\Gamma(m)  } \times \notag\\& \sum_{j=0}^{\infty} 
		\sum_{k=0}^j 
		\frac{\Gamma(k+m) K^j {}_2 F_1 \left(m+1,m+k,m+2,\frac{-m\alpha\gamma}{K+1}\right) }{j!k! (K+1)^m (m+1)}.
		\nonumber \\
% 		&\hspace{170pt}
		\end{align}
	\end{theorem}
	\begin{proof}
		See \textbf{Appendix C}.%\ref{apx:decoding}}.    
	\end{proof}

	\section{Resource Allocation for UAV-assisted D2D Underlaid Cellular Network}
	
	Based on the performance metrics characterized in the previous sections, in this section, we formulate and solve the link-type selection (i.e., direct D2D or relayed D2D), subchannel and power allocation problem for a UAV-assisted D2D underlaid cellular network.

	\subsection{Problem Formulation}
	{
		Let $\boldsymbol{\mu^r}\in\{0,1\}^{M^d}$, where $\mu_i^r=1$ and $\mu_i^r=0$ denote that $i$ is relayed or direct D2D pair, respectively. Besides, let $\boldsymbol{\rho}\in \{0,1\}^{M^d\times M^c}$, where $\rho_{i,j}=0$ and $\rho_{i,j}=1$ show that $i\in\M^d$ reuses or does not reuse the subchannel of $j\in\M^c$, respectively.} The link-type (direct or relayed), subchannel and power allocation optimization problem is formally stated as follows:
	%\begin{figure}
	\begin{align}
		\label{eq:optm_outer}
		({\bf P1}) \:
		\underset{\substack{ \pbold^d, \pbold^{dr}, \pbold^{rd}, \pbold^c }}{\mathrm{minimize}}
		\sum_i \left[   \mu_i^{r} (p_i^{dr}+p_i^{rd}) + (1-\mu_i^r) p_i^d \right] +
		\sum_{j}  p_j^c 
	\end{align}
	where $\pbold^d, \pbold^{dr}, \pbold^{rd}, \pbold^c$ belong to the set of all possible solutions of:
	\begin{subequations}
		\label{eq:optm_inner}	
		% 		\vspace{-30pt}
		\begin{align}
		\label{eq:optm_rate_objective}
		& \!\!\!\!\!\!\!\! 
		\underset{\substack{\boldsymbol{\rho}, \boldsymbol{\mu}^r,   \pbold^d, \pbold^{dr},  \pbold^{rd}, \pbold^c }}{\mathrm{maximize}}
		\sum\limits_{i} \sum\limits_{j}  (1-\mu_i^r)\rho_{i,j} R^{d,dir}(\zeta_{i,j}) + \nonumber
		\\ 
		%& \hspace{25pt}
		& \hspace{30pt} \sum\limits_{i} \sum\limits_{j_u} 
		\sum\limits_{j_d} \mu_i^r\rho_{i,j_u}\rho_{i,j_d} R^{d,rel}(\zeta_{i,j_u},\zeta_{i,j_d}) 
		\\
		% 		\hspace{15pt} &  \mathrm{subject\ to:}      \nonumber \\
		% 		\label{eq:optm_R}
		% 		&R^d=\sum\limits_{i} \sum\limits_{j}  (1-\mu_i^r)\rho_{i,j} R^{d,dir}(\zeta_{i,j}) + 
		% 		\sum\limits_{i} \sum\limits_{j_u} 
		% 		\sum\limits_{j_d} \mu_i^r\rho_{i,j_u}\rho_{i,j_d} R^{d,rel}(\zeta_{i,j_u},\zeta_{i,j_d}) 
		% 		\\
		\label{eq:optm_const_pout_cellular}
		  {\bf s.t.}& \  \rho_{i,j}\overline{\eps}_n(\zeta_{j,i}) \leq p_\eps, \hspace{10pt} \forall i, j, \\
		\label{eq:optm_const_Rcellular}
		& R^c(\zeta_{j,i}) \geq \rho_{i,j}  \widehat{R}^c, \hspace{10pt} \forall i, j, \\
		\label{eq:optm_const_rhod_domain}
		& \mu_i^r \in\{0,1\}, \rho_{i,j} \in\{0,1\}, \hspace{10pt} \forall i, j, \\
		\label{eq:optm_const_rho_sum_i_fixed_dir}
		& (1-\mu_i^r)\sum\nolimits_{j} \rho_{i,j} \leq 1, \hspace{10pt}  \forall i \\%| \mu_i^r=0, \\
		\label{eq:optm_const_rho_sum_i_fixed_rel}
		& 2(1-\mu_i^r)+ \mu_i^r \sum\nolimits_{j} \rho_{i,j} = 2, \hspace{10pt}  \forall i \\%| \mu_i^r=1, \\
		\label{eq:optm_const_rho_sum_j_fixed}
		& \sum\nolimits_{i} \rho_{i,j} \leq 1, \hspace{10pt}  \forall j, 
		\\
		\label{eq:optm_const_pmax} 
		& (p_i^d,p_i^{dr},p_i^{rd},p_j^c) \leq (\overline{p}_i^{d},\overline{p}_i^{dr},\overline{p}_i^{rd},\overline{p}_j^c), \hspace{10pt}  \forall i,j, 
		\\
		\label{eq:optm_const_pmin}
		& (\rho_{i,j}  \underline{p}_i^d,
		\rho_{i,j}  \underline{p}_{i,j}^{dr},
		\rho_{i,j}  \underline{p}_{i,j}^{rd},
		\rho_{i,j}  \underline{p}_{j,i}^c)
		\leq (p_i^d,p_i^{dr},p_i^{rd},p_j^c), \nonumber
		\\ 
		&\hspace{180pt}\forall  i,j. 
		%\nonumber
		%\\ 
		%&\hspace{180pt} \forall  i,j. 
		\end{align}
	\end{subequations}
	%\vspace{-30pt}
	%\end{figure}
	From \eqref{eq:optm_inner}, the objective of the {\bf inner problem} is to first  maximize the aggregate achievable data rate of all D2D pairs subject to the power constraints of  the all users and QoS (ergodic rate and reliability) constraints of cellular users. Also, \eqref{eq:optm_rate_objective} calculates the objective to be maximized as the total channel capacity of all relayed and direct D2D pairs.
	\eqref{eq:optm_const_pout_cellular} and \eqref{eq:optm_const_Rcellular} guarantee that the average decoding error probability and channel capacity of each cellular user $j$ are kept below the maximum allowed error probability $p_\eps$ and beyond the minimum allowed channel capacity $\widehat{R}^c$, respectively. \eqref{eq:optm_const_Rcellular} can be also viewed as the latency  constraint where $T_{j,i}=n_0/R^c(\zeta_{j,i})$ is the latency of cellular user $j$ if its channel is reused by D2D pair $i$ and $n_0$ is the total amount of data to be transferred, and similarly, $\widehat{T}=n_0/ \widehat{R}$ is the maximum acceptable latency of each cellular user. The constraints
	\eqref{eq:optm_const_rho_sum_i_fixed_dir}, \eqref{eq:optm_const_rho_sum_i_fixed_rel}, and \eqref{eq:optm_const_rho_sum_j_fixed} ensure that inactive direct D2D pairs are allocated to no subchannel, and each active direct D2D pair reuses the subchannel of one cellular user, and each active relayed D2D pair reuses the subchannels of two cellular users, respectively, and finally \eqref{eq:optm_const_pmax} and \eqref{eq:optm_const_pmin} guarantee that the power constraints of all D2D pairs and cellular users hold. {Note that, \eqref{eq:optm_const_pmin} guarantees that the interference-limited assumption always hold.}
	
	To solve the considered MINLP problem, our  methodology is  as follows:
	\begin{enumerate}
		\item For a given subchannel assignment, we first rewrite {\bf P1}  as shown in {\bf P2} of \textbf{Section~V.B} to optimize power allocations. We then solve the inner problem of {\bf P2} to get closed-form optimal power allocations  by reformulating the power variables in terms of new variable $\eta=p_i/p_j$.  The optimal value of $\eta$ (denoted by $\eta^*$) corresponds to all optimal points ($p_i$,$p_j$) lying in the line $p_i=\eta^* p_j$ (provided that the feasibility and interference limited assumptions hold, which is verified in the proof of \textbf{Theorem~2}).
		% \item We first calculate the {\em closed-form optimum power solution and rate} for a given subchannel which is shared between one D2D pair and one cellular user ({\bf Section V.B}).
		\item Then, from the space of  optimal powers obtained from the  {\bf inner problem of P2}, we derive a power allocation solution corresponding to the minimum aggregate transmit power of {the UAVs and} all D2D and cellular devices, as in the {\bf outer problem}  of {\bf P2}. 
		\item Then, we formulate our problem as one-to-many matching problem which is solved by first obtaining the {\em optimal subchannel and power allocation} of the maximum weighted one-to-one matching problem (assuming no UAV relays) and then we extend the scheduling scheme  to one-to-many matching problem with UAV relays ({\bf Section V.C}).
	\end{enumerate}

	%   \begin{remark}
	%   	   The considered optimization problem  in \eqref{eq:optm_outer} and \eqref{eq:optm_inner} is different from the weighted bi-objective optimization where a weighted combination of two or more objectives is considered as a single objective and then solves the whole problem as a single optimization problem. Here, we first obtain all possible solutions of the inner problem and then selects the one corresponding to the optimal point of the outer problem. {Clearly,  if the solution space of the inner problem is a single point (which is not the case in our problem as will be shown later), the outer problem  has no effect and simply obtains the same point as the inner one.}
	%   \end{remark}
	Our solution approach is different from iterative alternating optimization methods where a problem is typically split into several sub-problems and each sub-problem solves only one optimization variable given the remaining optimization variables  from the previous iteration.

	\subsection{Optimal Power Allocations for D2D and cellular users}
	In the following, we present the closed-form power allocations for any possible subchannel assignment. 
	Suppose that D2D pair $i\in\M^d$ reuses the subchannel allocated to cellular user $j\in\M^c$. 
	Let $p_j\stackrel{\Delta}{=}p_j^c$ be the transmit power of cellular user $j$. Without loss of generality, we assume that $i$ is a direct D2D pair with transmit power $p_i \stackrel{\Delta}{=}p_i^d$, however one can easily verify that the following analysis is also valid for the uplink of relayed D2D pair $i$ with transmit power $p_i \stackrel{\Delta}{=}p_i^{dr}$, and for the downlink of relayed D2D pair $i$ with transmit power $p_i \stackrel{\Delta}{=}p_i^{rd}$.
	The power allocation for D2D pair $i$ and cellular user $j$ corresponding to \eqref{eq:optm_outer} and \eqref{eq:optm_inner} is stated as follows:
	\begingroup
	\begin{align}
	\label{eq:optm_outer_reused}
	({\bf P2}) \qquad 		 \underset{\substack{ p_i , p_j }}{\mathrm{minimize}}\ p_i + p_j \hspace{110pt}
	\end{align}	
	where $p_i, p_j$  belong to the set of all possible solutions of: %\hspace{97pt}
	\begin{subequations}
		\label{eq:optm_inner_reused}	
		\begin{align}
		\label{eq:optm_rate_objective_reused}
		\underset{\substack{p_i , p_j }}{\mathrm{maximize}}\ & R^{d,dir}(\zeta_{i,j}) \hspace{50pt}  \\
		% 			\label{eq:optm_const_pout_cellular_reused}
		% 				\mathrm{s.\ t.}\ &	O^{cel}(\zeta_{j,i}\gammatarget_j) \leq p_0, \\
		\label{eq:optm_const_pout_cellular_reused}
		\mathrm{s.\ t.}\ &	R^c(\zeta_{j,i}) \geq \widehat{R}^c, \\
		\label{eq:optm_const_Rcellular_reused}
		&\overline{\eps}_n(\zeta_{j,i})\leq p_\eps  , \\
		\label{eq:optm_const_p_reused}
		&(\underline{p}_i, \underline{p}_j) \leq  (p_i,p_j) \leq (\overline{p}_i,\overline{p}_j),
		\end{align}
	\end{subequations}
	\endgroup
	where $\underline{p}_i\stackrel{\Delta}{=}\underline{p}_{i,j}$ and $\underline{p}_j\stackrel{\Delta}{=}\underline{p}_{j,i}$ are obtained from \eqref{eq:interference_limited_dir_k} and $\overline{\eps}_n(.)$ is obtained from \eqref{eq:eps_n_average}.

	\begin{theorem}
		\label{th:power}
		The solution to power allocation sub-problem given in (\ref{eq:optm_outer_reused}) and (\ref{eq:optm_inner_reused}) is:
		\begin{align}
		\label{eq:opt_reused_subproblem_power_pair}
		(p^{\hspace{-1pt}*}_{i} , p^{\hspace{-1pt}*}_{j})=
		\begin{cases}
		\left( \eta^* \underline{p}_j  , \underline{p}_j \right) ,  &\hspace{-5pt}
		\mbox{if } \underline{p}_i/\underline{p}_j \leq \eta^*, % \leq \overline{p}_i/\underline{p}_j
		\\
		\left( \underline{p}_i ,  \underline{p}_i / \eta^* \right), &\hspace{-5pt}\mbox{if } \underline{p}_i/\overline{p}_j \leq \eta^* < \underline{p}_i/\underline{p}_j,
		\\
		\emptyset, &\hspace{-5pt}\mbox{if } \eta^* < \underline{p}_i/\overline{p}_j,
		\end{cases}
		\end{align}
		where $\eta^*= \min{ \{ \eta_1^*, \eta_2^*,  \overline{p}_i/\underline{p}_j \}}$ in which $\eta_1^*$ and $\eta_2^*$ are unique solution values of $\eta$
		in the equality constraints $R^c\left({\eta  \mathbb{E}[\widehat{h}_{j,i}]/  \mathbb{E}[h_j]}\right)=\widehat{R}^c$ and $\overline{\eps}_n\left({\eta { \mathbb{E}[\widehat{h}_{j,i}]}/  {\mathbb{E}[h_j]}}\right)=p_\eps$, respectively.
	\end{theorem}
	\begin{proof}
		\begin{figure*}
		\centering
			\includegraphics [width=450pt, height=135pt]{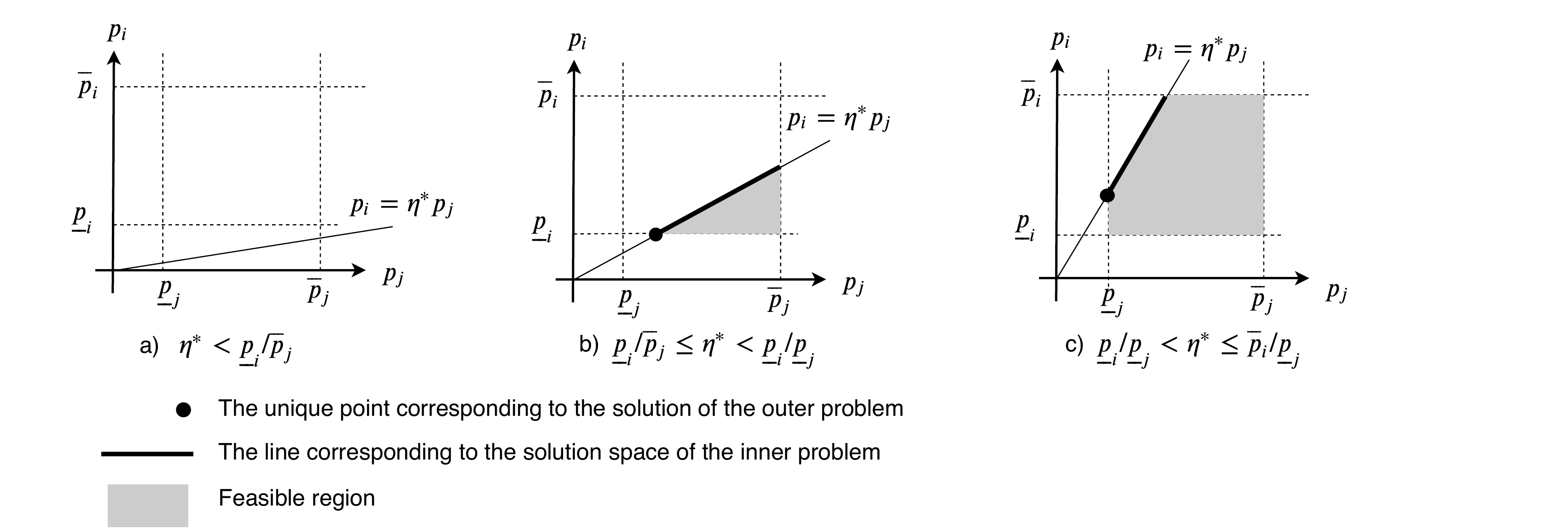} 
			\caption{Illustration of three possible cases of the solution of inner and outer problems. The thick lines show the solution space of the inner problem, and the small filled circles show the single point corresponding to the solution of the outer problem.} \vspace{-20pt}
			\label{fig:optimal_cases}
		\end{figure*}
		We first find the solution to the inner optimization sub-problem (\ref{eq:optm_inner_reused}). By considering $\eta={p_i}/{p_j}$ as the optimization variable, (\ref{eq:optm_inner_reused}) can be rewritten as
		\begin{subequations}
			\label{eq:optm_inner_reused2}
			\begin{align}
			\label{eq:optm_rate_objective_reused2}
			\underset{\substack{\eta }}{\mathrm{maximize}}\ & R^{d,dir}\left( k_1 / \eta \right), \hspace{100pt}  \\
			\label{eq:optm_const_pout_cellular_reused2}
			\mathrm{s.\ t.}\ &	R^c(k_2  \eta) \geq \widehat{R}^c, \\
			\label{eq:optm_eps}
			&\overline{\eps}_n(k_2  \eta) \leq  p_\eps, \\
			\label{eq:optm_const_eta_reused}
			&\underline{p}_i / \overline{p}_j \leq  \eta \leq \overline{p}_i /\underline{p}_{j},
			%&\eta^{min} \leq  \eta \leq \eta^{max},
			\end{align}		
		\end{subequations}
		where $k_1=\mathbb{E}[\widehat{h}_{i,j}]/\mathbb{E}[h_i]$, and $k_2= \mathbb{E}[\widehat{h}_{j,i}]/  \mathbb{E}[h_j]$.
		First, we prove that $\overline{\eps}(k_2 \eta)$ in \eqref{eq:optm_eps} is an increasing function of $\eta$. 
		From \eqref{eq:eps_n_average} and \eqref{eq:H}, we have
		\begin{align*}
		&\overline{\eps}_n(k_2\eta)
		=\sum_{i=1}^{L-1}(\omega_{i-1}-\omega_i)H(k_2\eta,\gamma_i) + \omega_L H(k_2\eta,\gamma_L)
		\\
		&=
		\sum_{i=1}^{L-1}(\omega_{i-1}-\omega_i)\int_0^{\gamma_i} \!\! O^c(k_2 \eta \gamma) d\gamma + \omega_L \int_0^{\gamma_L} \!\! O^c(k_2 \eta \gamma) d\gamma.
		\end{align*}
		Noting that $\omega_{i-1}>\omega_{i}, \forall i$, and also by noticing that $O^c(k_2 \eta\gamma)=\mathrm{Pr}\{\frac{y}{x} \leq k_2 \eta \gamma \}$ is an increasing function of $\eta$ (where $y$ and $x$ are fading variables associated with the channel of cellular user $j$ and the interfering channel from user $i$), it is clear that $\overline{\eps}_n(k_2\eta)$ is an increasing function of $\eta$.
		Hence, we can equivalently replace constraint \eqref{eq:optm_eps} by $\eta \leq \eta_2^*$. Now we prove that $R^c(k_2/\eta)$ is also a monotonically increasing function of $\eta$.	From \eqref{eq:R_c_dir} and the definition of $f_r(.)$, we have
		\begin{multline}
    		\label{eq:tt1}
    		R^c(k_2 \eta) = \frac{1}{\mathrm{ln}(2)}\int_{0}^{\infty} \frac{1-O^c(k_2\gamma \eta)}{1+\gamma} d\gamma %\nonumber \\
    		=
    		\\
    		\frac{1}{\mathrm{ln}(2)}\int_{0}^{\infty} \frac{1 -O^c(\gamma)}{k_2\eta +\gamma} d\gamma.
		\end{multline}
		Noting the fact that $O^c(\gamma) \leq 1$, from (\ref{eq:tt1}) we conclude that $R^c(k_2 \eta)$ is an decreasing function of $\eta$ and so, \eqref{eq:optm_const_pout_cellular_reused2} can be replaced by $\eta \leq \eta_1^*$.
		In a similar way, we could verify that $R^{d,dir}(k_1/\eta)$ is an increasing function of $\eta$. Thus, we can rewrite (\ref{eq:optm_inner_reused2}) as
		\begin{align}
		\label{eq:optm_rate_objective_reused3}
		\underset{\substack{\eta }}{\mathrm{maximize}}\ & R^{d,dir}\left( k_1 / \eta \right) \nonumber  \\
		& \underline{p}_i / \overline{p}_j \leq  \eta \leq \min{ \{ \eta_1^*, \eta_2^*,  \overline{p}_i/\underline{p}_j \} },
		\end{align}		
		and hence we conclude that $\eta^*= \min{ \{ \eta_1^*, \eta_2^*, \overline{p}_i/\underline{p}_j \}}$ is the unique solution to (\ref{eq:optm_inner_reused2}) provided that $\eta^* \geq \underline{p}_i / \overline{p}_j$ ; otherwise the problem is infeasible.  
		%Thus, all feasible points lying in the line $p_i=\eta^* p_j$ (provided that $\eta^* \geq \underline{p}_i / \overline{p}_j$. 
		Fig. \ref{fig:optimal_cases} shows that depending on $\eta^*$ (where $0<\eta^*\leq \overline{p}_i/\underline{p}_j$), three possible cases may happen. For case (a) where $\eta^*<\underline{p}_i/\overline{p}_j$, there exists no feasible point. The set of optimal solutions of the inner problem for cases (b) and (c)  are illustrated as the thick lines corresponding to the intersection of the feasible region (displayed in gray) and the line $p_i=\eta^* p_j$. Now, referring to problem (\ref{eq:optm_outer_reused}) and  constraint \eqref{eq:optm_const_p_reused}, the optimal power allocation %solution for power allocation of the coupled D2D and cellular users are obtained as
		is obtained as
		(\ref{eq:opt_reused_subproblem_power_pair}), which is shown in black filled circles in the figure.
	\end{proof}

	\subsection{Link-type, Subchannel, and Power Allocation for D2D Pairs}
	In order to solve the link-type, subchannel and power allocation problem stated in (\ref{eq:optm_outer}) and (\ref{eq:optm_inner}), first we consider the case where no UAV relay is involved and all D2D pairs can only establish direct links. 
	%Here the problem can be modeled as a one-to-one maximum weighted matching game. 
	As seen in Fig. \ref{fig:matching}-a, consider a bipartite graph composed of two sets of vertices, namely D2D pairs and cellular users whose subchannels are going to be reused by D2D pairs. Each D2D pair can only reuse the subchannel of one cellular user and each cellular user can share its subchannel to only one D2D pair, thus the {scheduling is a one-to-one maximum weighted matching problem which can optimally be solved by Hungarian algorithm. The procedure of calculating the powers and weights for each link and obtaining the optimal solution is described in {\bf Algorithm 1}}.
	
	\begin{algorithm}[t]
		\caption{\small\!: Joint Optimal Power and Subchannel Allocation for Direct D2D Links}
		\begin{algorithmic}[1]
			\State \textbf{Initialize:} Set $w_{i,j}=0, \forall i\in\M^d,\forall j\in\M^c$;
			\For{\textbf{each} $i\in\M^d$} 
			\For{\textbf{each} $j\in\M^c$}
			\State Obtain $p_{i,j}=(p_i^*,p_j^*)$ from  (\ref{eq:opt_reused_subproblem_power_pair}). If $p_{i,j} \! \neq \! \emptyset$ then $w_{i,j} \! = \! R^{d,dir} \!\left( \! \zeta_{i,j} \big|_{\substack{p_i=p_i^*, p_j=p_j^*}}\!\right)\!;$
			\EndFor
			\EndFor
			\State   Find the optimal scheduling corresponding to maximum  weighted matching of $\boldsymbol{w}$ through Hungarian algorithm.
		\end{algorithmic}
	\end{algorithm}
	%\end{figure}
	\begin{figure*}
		\centering
		\includegraphics [width=464pt]{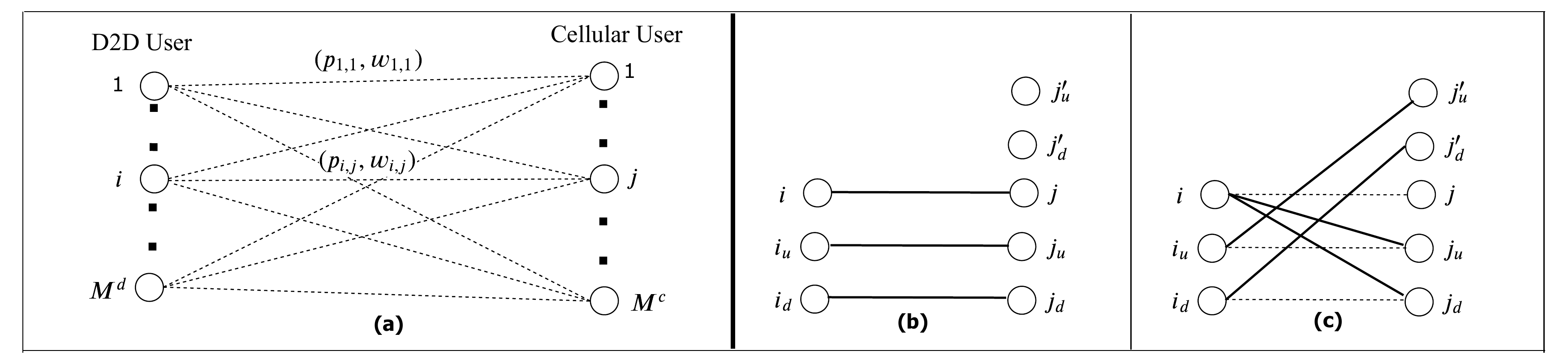} \\
		\caption{Bipartite graph for the matching problem. Fig. (a) shows the one-to-one matching for Algorithm 1 for the case when no UAV relay is used. Figs (b) and (c) show the procedure of changing a direct D2D link to a relayed link in Algorithm 2.} \vspace{-10pt}
		\label{fig:matching}
	\end{figure*}
	In order to further increase the throughput of D2D links, we then update Algorithm 1 by replacing  direct D2D links with poor QoS with relayed links. This problem is modeled as a one-to-many matching game, where each D2D pair may reuse the subchannel of one cellular user (if it establishes a direct D2D link) or the subchannels of two cellular users (for the case of relayed D2D links), but each cellular user still shares its subchannel to at most one D2D pair. Let $i=\rho^d(j)$ denotes a D2D pair $i$ reusing the subchannel of cellular user $j$,  $j=\rho^c(i)$ denotes a cellular user $j$ whose subchannel is reused by D2D pair $i$ in the direct link mode, and $(j_u,j_d)=\rho^c(i)$ denotes cellular users $j_u$ and $j_d$ whose subchannels are reused respectively in the uplink and downlink paths of the relayed link of D2D pair $i$. The relay-assisted scheduling and power allocation problem is proposed in {\bf Algorithm 2}.
	\begin{algorithm}[]
		\caption{\small\!:  Joint Link-Type, Subchannel, and Power Allocations for problem ({\bf P1})}
		\begin{algorithmic}[1]
			\State %\textbf{Initialize:}
			 Find the best matching for direct D2D links through Algorithm 1; $\M^{d,t}\leftarrow \M^d$; $\M^{c,t}\leftarrow \M^c$; 
			\For{\textbf{each} $i\in\M^{d,t}$}
			%\For{\textbf{each} $j_u\in\M^{c,t}$}
			\For{\textbf{each} $j_u\in\M^{c,t},j_d\in\M^{c,t}$ where $j_d\neq j_u$}
			\State $j=\rho^c(i)$; $i_u=\rho^d(j_u)$; $i_d=\rho^d(j_d)$;
			\State For pair $(i,j_u)$ obtain $(p_i^{*dr},p_{j_u}^*)$, and for pair 
			%\Statex \hspace{55pt}
			$(i,j_d)$
			obtain $(p_i^{*rd},p_{j_d}^*)$ from  (\ref{eq:opt_reused_subproblem_power_pair});
			%      \State
			%      $w_{i,j_u, j_d} \! = \! R^{d,rel} \!\left( \! \zeta_{i,j_u} \Bigg|_{\substack{p_i=p_i^{*rd}\\p_j=p_{j_u}^*}},
			% 			\zeta_{i,j_d} \Bigg|_{\substack{p_i=p_i^{*dr}\\p_j=p_{j_d}^*}}
			% 			\right)\!$;
			\State
			$w_{i,j_u, j_d} \! = \! R^{d,rel} \!\left( \! \zeta_{i,j_u} \Bigg|_{\substack{p_i=p_i^{*rd},\\ p_j=p_{j_u}^*}},
			\zeta_{i,j_d} \Bigg|_{\substack{p_i=p_i^{*dr},\\p_j=p_{j_d}^*}}
			\right)\!$;
			\State
			$(j'_u,j'_d)=\underset{j_1,j_2 \in\M^{c,t}\setminus \{j_u,j_d\}}{\argmax} \{ w_{i_u,j_1} + $
			%\Statex  \hspace{55pt}
			$w_{i_d,j_2} \big|   j_1\neq j_2, \rho^d(j_1)=\rho^d(j_2)=0 \}$ ;
			%\State
			%Save $j'_u$ and $j'_d$ as candidates that $i_u$ and $i_d$
			%\Statex \hspace{55pt}
			%may reuse their subchannels to build new 
			%\Statex \hspace{55pt}
			%direct links, i.e., $j^{replace}_{i,j_u,j_d}=(j'_u,j'_d)$;
			\State Save $j'_u$ and $j'_d$ as candidates that $i_u$ and $i_d$ may reuse their subchannels to build 
			\Statex \hspace{70pt}
			new direct links, i.e., $j^{replace}_{i,j_u,j_d}=(j'_u,j'_d)$;
			\State
			$\Delta w_{i,j_u,j_d}=w_{i,j_u, j_d}- w_{i,j} - w_{i_u,j_u} - w_{i_d,j_d} +$ 
			%\Statex \hspace{55pt}
			$w_{i_u,j'_u} + w_{i_d,j'_d}$;
			\EndFor
			\EndFor
			%\EndFor
			\State $(i^*,j_u^*,j_d^*)=\underset{i,j_u,j_d}{\argmax\ } \left\{ \Delta w_{i,j_u,j_d} \big| \Delta w_{i,j_u,j_d} > 0 \right\}$;
			\If {$(i^*,j_u^*,j_d^*)\neq \emptyset$}
			\State $j^*=\rho^c(i^*)$; $i_u^*=\rho^d(j_u^*)$;\ $i_d^*=\rho^d(j_d^*)$;
			
			\State
			$(j_1^*,j_2^*)\leftarrow j^{replace}_{i^*,j_u^*,j_d^*}$;
			
			\State
			$\rho^c(i^*) \leftarrow (j_u^*,j_d^*)$;
			$\rho^c(i_u^*) \leftarrow j_1^*$; $\rho^c(i_d^*) \leftarrow j_2^*$;
			
			\State $\M^{d,t} \leftarrow \M^{d,t} / \{i^*\}; \M^{c,t} \leftarrow \M^{c,t} / \{j_u^*, j_d^*\} $;
			
			\State Go to step 2;
			\Else
			\State Scheduling is terminated;
			\EndIf
		\end{algorithmic}
	\end{algorithm}

	As seen in {\bf Algorithm 2},  first we find the optimal scheduling for direct D2D links establishment through Hungarian algorithm, and then we form the set of temporary candidates of relayed D2D pairs as $\M^{d,t}$ and cellular users whose subchannels can be reused by relayed D2D candidates as $\M^{c,t}$. In order to see how relayed links are established, from Fig. \ref{fig:matching}-b, suppose $j,j_u, j_d\in\M^{d,t}$ have initially shared their subchannels to direct D2D pairs $i,i_u,i_d$ respectively. We denote $i_u=\emptyset$, and $i_d=\emptyset$ for the  case when $j_u$ or  $j_d$ have shared their subchannels to none of D2D links. In Steps 2-13, we search among all possible allocations of $j_u,j_d\in\M^{c,t}$ to establish a relayed connection to some $i\in\M^{d,t}$ which results in the maximum throughput increase $\Delta w_{i,j_u,j_d}$. In assigning the subchannels of $j_u$ and $j_d$ to D2D pair $i$, as seen in Step 10, and Fig. \ref{fig:matching}-c, the terms $w_{i,j}$, $w_{i_u,j_u}$, and $w_{i_d,j_d}$ are deducted from $\Delta w_{i,j_u,j_d}$ since their corresponding connections are lost, and instead, the terms $w_{i,j_u,j_d}$, $w_{i_u,j'_u}$, and $w_{i_d,j'_d}$ are added. Note that, in order to replace the direct links of $i_u$ (if $i_u\neq\emptyset$) and $i_d$ (if $i_d\neq\emptyset$) from $j_u$ and $j_d$, to $j'_u$ and $j'_d$, as seen in Step 8, we search through all possible candidates that may maximize the sum of the weights of direct links $w_{i_u,j'_u}+w_{i_d,j'_d}$. After finding the best candidate $(i^*,j_u^*, j_d^*)$ in Step 14, we set all required allocations in Steps 16-18, and remove $\{i^*\}$ and $\{j_u^*,j_d^*\}$ from the candidate sets $\M^{d,t}$ and $\M^{c,t}$, respectively in Step 19, and then we continue the procedure again until no allocation can be found that improves $\Delta w$.

	As a benchmark for the proposed Algorithm~2, we consider two versions of low complexity greedy algorithms. In the first algorithm denoted by \textit{Greedy1}, we first obtain the optimal subchannel and power solution through Algorithm 1 and then D2D pairs sequentially change their direct links to relayed links (in a greedy manner by selecting uplink and downlink channels resulting in the highest throughput) provided that the corresponding user's throughput is increased by changing its link-type. In the second version denoted by \textit{Greedy2}, after obtaining the optimal solution of Algorithm 1, we first sort D2D pairs in ascending order according to the achieved direct-link data rates, aiming at D2D pairs with low data rates to opt for higher quality relayed links, and then a similar mechanism as \textit{Greedy1} is employed.
	
	\subsection{Complexity Analysis}
	For one-to-one matching game, described in {\bf Algorithm 1}, the standard Hungarian algorithm can be solved with time-complexity of $O((M^d)^2 M^c)$ \cite{west1996introduction}.    Regarding the complexity of {\bf Algorithm 2}, in order to make the analysis simpler, we calculate the worst case for the upper-bound time-complexity. The steps 5-10 are executed by $O(1)$. More specifically, we show this for Step 8 (which is potentially of most complexity among others). After finishing the initialization step and before starting Step 2, for each $i$ we can create the vector of elements $w^{sorted}_i(j)$ by sorting $w_{i,j}$ for all $j$ such that $\rho(j)=0$, and thus, by doing this for all $i$, we can create the sorted weight matrix $w^{sorted}$ with an overall complexity of $O(M^c M^d)$. This way, Step 8 can always be executed by $O(1)$; Hence the overall complexity of each iteration for Steps 2-13 in the worst case (by ignoring the user removals from $\M^{c,t}$, i.e., by assuming $|\M^{c,t}|=M^c$) is $O(M^d (M^c)^2)$. Now, since the procedure is run again, in Step 20, after the removal of $i^*$ from $\M^{d,t}$ in Step 19, the complexity is obtained as $O((M^d M^c)^2)$. Thus, by considering the execution of Hungarian Algorithm in Step 1, the overall upper-bound complexity  will be $O((M^d M^c)^2)+O((M^d)^2 M^c)=O((M^d M^c)^2)$. 
	Similar to the discussion presented above, the complexity of \textit{Greedy1} Algorithm is easily obtained as $O((M^d)^2 M^c) + O(M^d (M^c)^2)$. Complexity of \textit{Greedy2} Algorithm is similar to that of \textit{Greedy1} added by the complexity of sorting D2D rates; thus it is found to be $O(M^d\log(M^d))+O((M^d)^2 M^c) + O(M^d (M^c)^2)=O((M^d)^2 M^c) + O(M^d (M^c)^2)$. 
	% A comparison of the complexity of proposed algorithms is found in Table \ref{tbl:cmplexity}.
	% 	\begin{table}[]
	% 		\centering
	% 		\caption{Comparison of the complexity of proposed algorithms}
	% 		\vspace{-10pt}
	% 		\label{tbl:cmplexity}
	% 		\begin{tabular}{|c|c|c|c|}
	% 			\hline
	% 			\textbf{Algorithm 1} & \textbf{Algorithm 2} & \textbf{Greedy1} & \textbf{Greedy2} \\
	% 			\hline
	% 			%a & Algorithm 2 &  &  \\
	% 			$O(M^d M^c)$ & $O((M^d M^c)^2)$ & $O((M^d)^2 M^c) + O(M^d (M^c)^2)$  &  $O((M^d)^2 M^c) + O(M^d (M^c)^2)$\\
	% 			\hline
	% 		\end{tabular}
	% 			\vspace{-20pt}
	% 	\end{table}     
	%

	\section{Simulation Results}	
	%In this section we present different simulation results to analyze the performance of proposed algorithms versus various parameters. 
	Consider a cell where the BS is located at the center and several cellular and D2D users are randomly located within the cell. Simulation parameters are listed in Table \ref{tbl:simulation_params}. The network frequency and path-loss parameter values are considered according to the experimental results given in \cite{akdeniz2014millimeter}. In all scenarios, except for Fig. \ref{fig:sim_d} in Section \ref{sec:sim_uav} (wherein simulation measures are obtained versus fixed values of D2D pairs distances), cellular users as well as D2D transmitters and receivers are randomly located within the cell area. Based on the positions of the users and UAVs, the probability of LoS relating to main and interference signals of  aerial and terrestrial links are obtained for each user according to \eqref{eq:p_LoS_per_teta} and \eqref{eq:p_LoS_per_d} respectively. % in a way that the link length (distance between transmitter and receiver) is a uniform random value between 10 m and 600 m.
	By default, we consider two UAVs located at $(450,0)$ and $(-450 ,0)$ with respect to the center, and the height of UAVs are considered to be $300$ m. In order to compensate for the severe path-loss in mm-wave frequencies, it is needed to employ beamforming in the antennas of transmitters and receivers. \textcolor{black}{For each transmitting antenna, we consider that the beamforming is applied in a way that the maximum directivity of the antenna is steered toward the desired link. Therefore, the main receiver receives the desired signal with maximum directivity gain $A^{max}$, and the interfered receiver is subject to a directivity gain lower than $A^{max}$. Let $\theta$ be the angle between a desired link and its corresponding interfering link. We model the antenna beamforming directivity pattern similar to the Gaussian-like shape proposed and employed in \cite{martin2009algorithms} and \cite{werner2015performance} as $
		A(\theta)=\max\left\{A^{max}\exp\left(-0.69\frac{\mathcal{N}(\theta)^2}{\theta_{\textrm{3dB}}^2}\right),A^{min} \right\}
		$
		where $\mathcal{N}(\theta)=\mod_{2\pi}(\theta+\pi)-\pi$, %$\theta$ is the angle between the main receiver and the interfered receiver,
		$A^{max}$ is the maximum gain (which is radiated toward the main receiver), $A^{min}$ is a minimum constant, and $\theta_{3\textrm{dB}}$ is the half-power beamwidth (i.e., $A(\theta_{3\textrm{dB}})={A^{max}}/{2}$)}.
	In what follows, simulation parameter values are taken from Table \ref{tbl:simulation_params}, unless explicitly stated otherwise. The results for all following figures are obtained by Monte-Carlo.
	\begin{table*}
		\centering
		\caption{Simulation Parameters}
		\begin{tabular}{|l|l|l|l|}
			\hline
			\textbf{Parameter}                                & \textbf{Description} 
			&\textbf{Parameter}                                & \textbf{Description} 
			\\
			\hline
			Frequency 						& $28$ GHz      &Noise power					& $1.1 \times 10^{-12}$ Watts
			\\						
			Average path-loss for NLoS links		& $72+29.2\log_{10}(d)$ dB     
			&($M^c, M^d$)	                & $(18,10)$    \\
			Average path-loss for LoS links 		& $61.4+20\log_{10}(d)$ dB     
			&Cell area						& $1600\times 1600 m^2$                       
			\\
			($\overline{p}_{i}^d,\overline{p}_{i}^{dr},\overline{p}_{i}^{rd},\overline{p}_{j}^c$)			& ($0.1,0.1,1,0.1$) watts 
			&$\widehat{R}^c$			& $8$ bps/Hz 
			\\
			$(B,C)$ in \eqref{eq:p_LoS_per_teta} \cite{mozaffari2016unmanned}					& (0.1396,11.95) 
			&Fading factors $(m,K)$								& $(2, 12 \textrm{ dB})$ 
			\\
			%$\widetilde{K}$					& 20 \\	
			UAV height, BS height  & $250$ m, $50$ m
			&	$(A^{max},\theta_{3\textrm{dB}})$
			& ($25$ dB,  $15^\circ$)
			\\
			$(p_{\eps},k_n)$			& $(10^{-4},5)$ 
			& & %$d^{max}$ in \eqref{eq:p_LoS_per_teta}								& 1000 m 
			\\			
			\hline
		\end{tabular}
		\label{tbl:simulation_params}
	\end{table*}
	\subsubsection{Validation of derived expressions}
	\begin{figure*}
		\begin{minipage}{.47\linewidth}
			\centering
			\includegraphics [width=230pt,height=105pt]{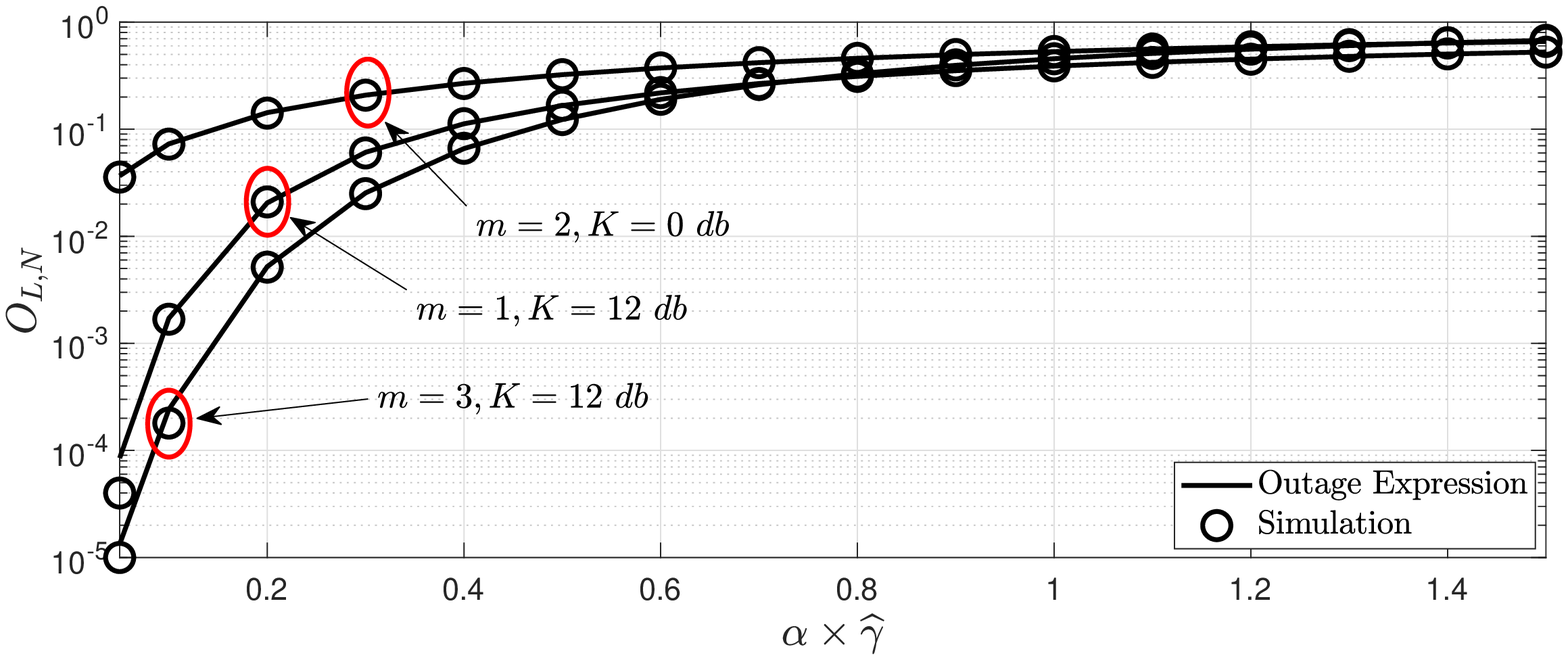} 
			\includegraphics
			[width=230pt,height=105pt]{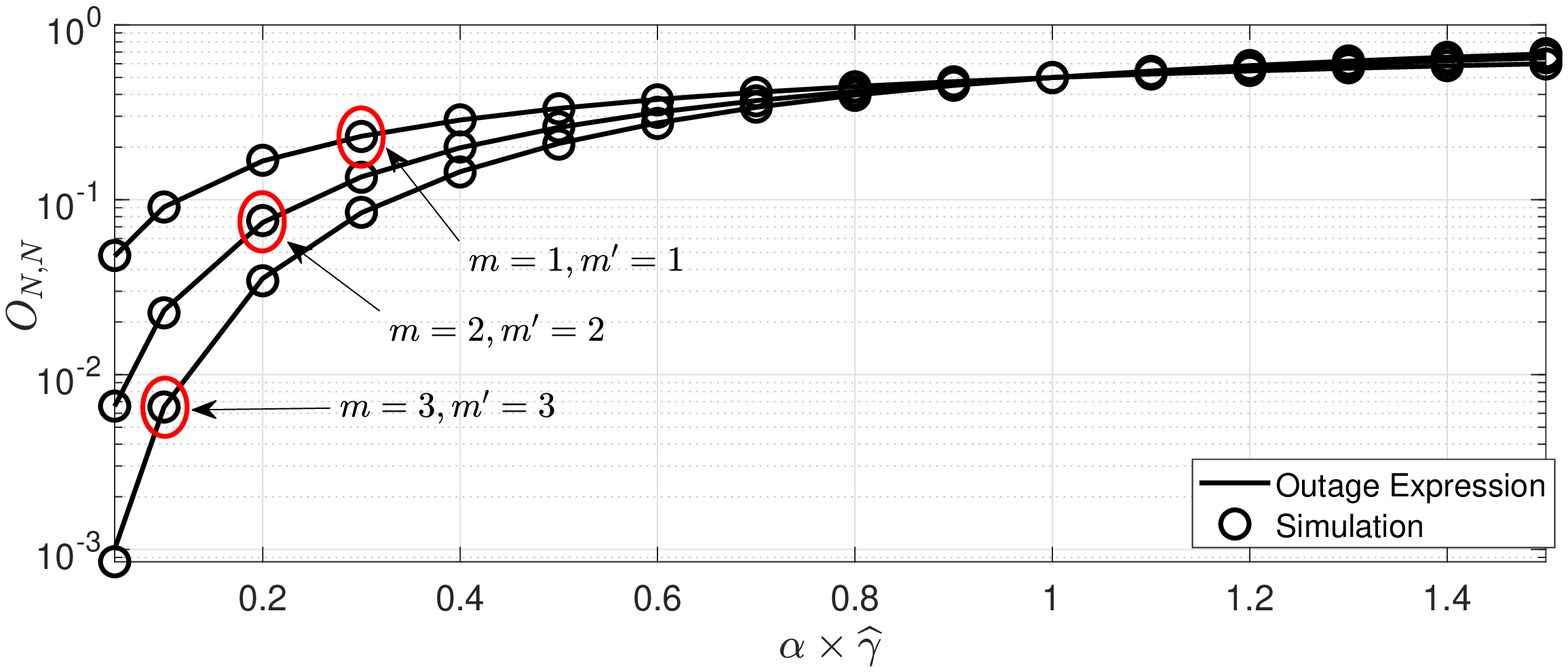}	
			\caption{Comparison of outage probability of derived expressions in \eqref{eq:outage_L_N} (first figure) %and \eqref{eq:outage_N_L} (second figure) 
				and \eqref{eq:outage_N_N} (second figure) and their corresponding values obtained through simulation.} \vspace{-10pt}
			\label{fig:sim_outage_verify}
		\end{minipage}
		\hspace{10pt}
		\begin{minipage}{.47\linewidth}
			\centering
			\includegraphics [width=230pt,height=105pt]{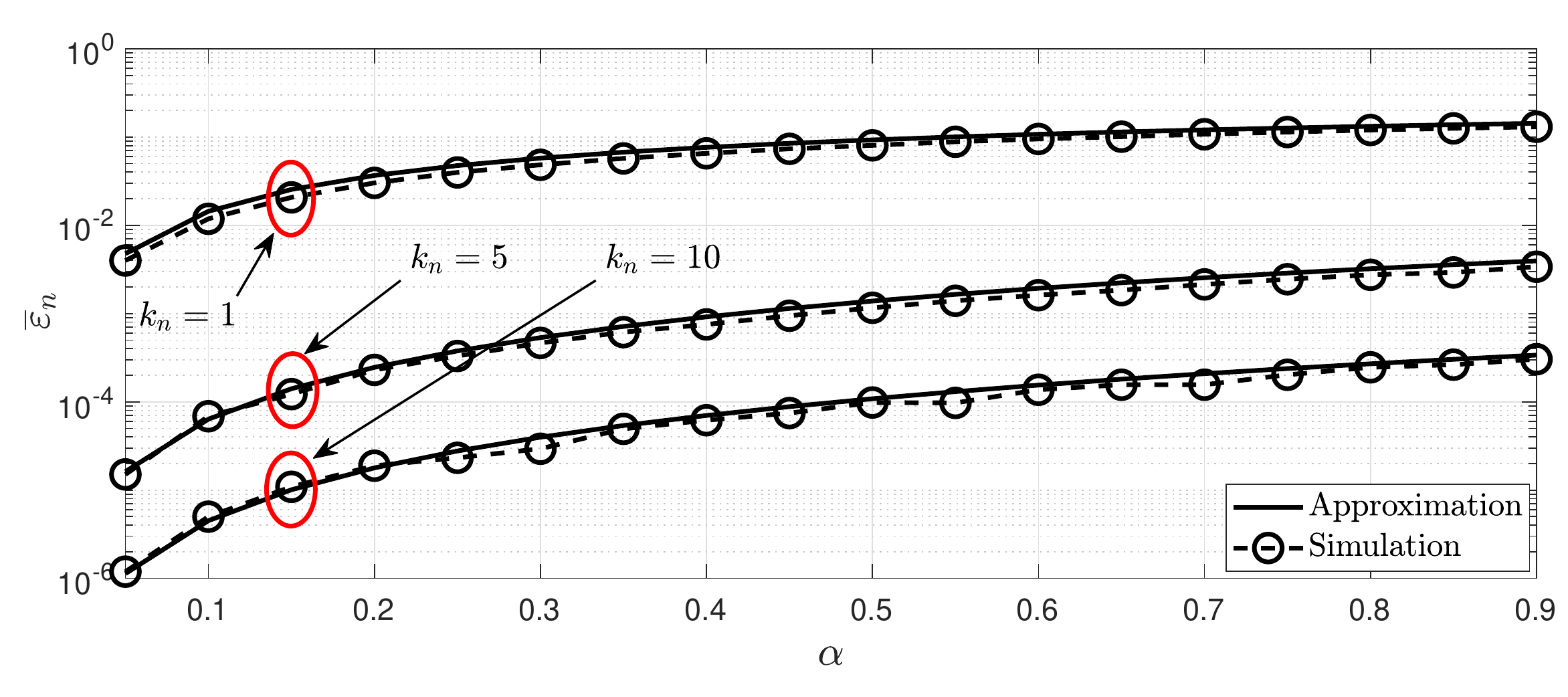} 
			\includegraphics
			[width=230pt,height=105pt]{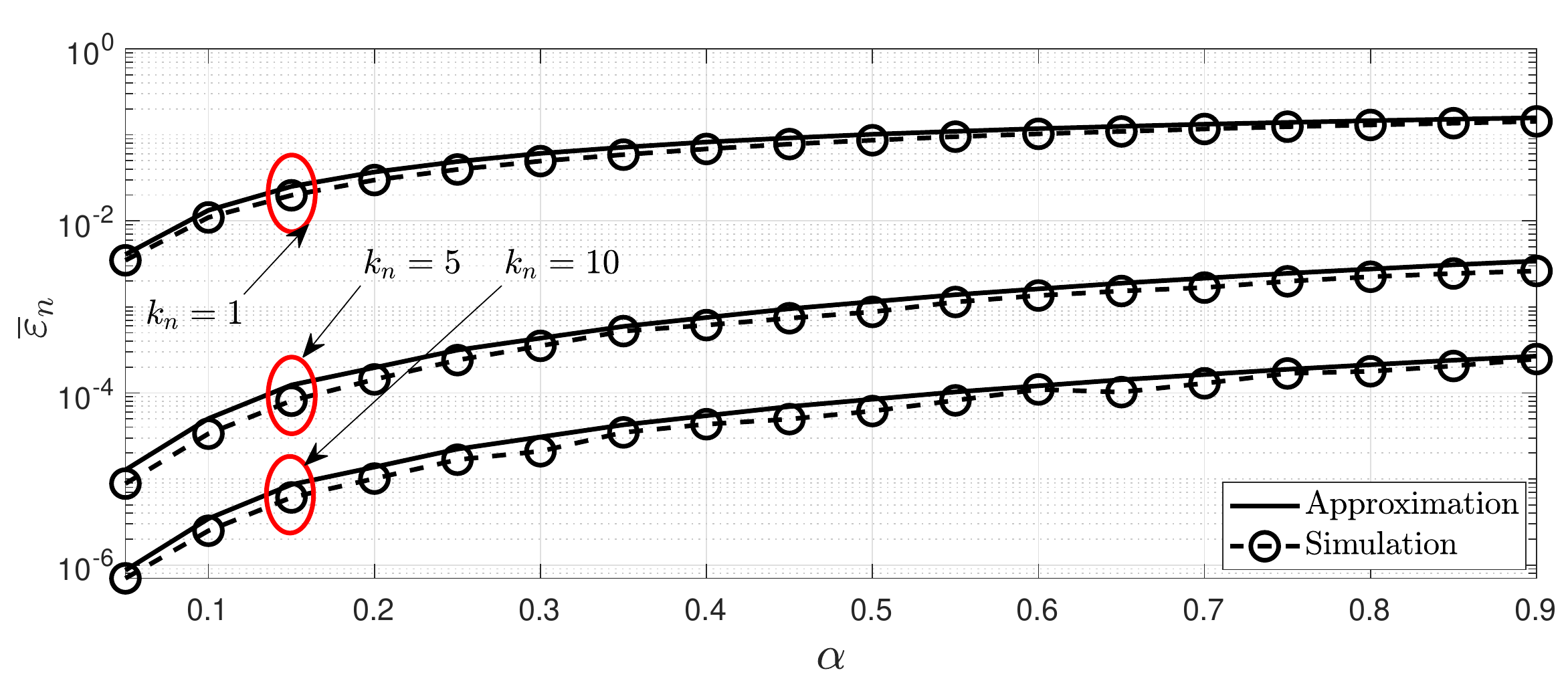} 
			\caption{Comparison of the  frame decoding error probability of \eqref{eq:H_N} (first figure) and \eqref{eq:H_L} (second figure) and corresponding exact values obtained from simulations.} \vspace{-10pt}
			\label{fig:sim_epsilon_verify}
		\end{minipage}
	\end{figure*}

	In order to validate the derived expressions, Fig. \ref{fig:sim_outage_verify} depicts the outage probability expressions  
	of $O_{L,N}$ in \eqref{eq:outage_L_N}, 
	%$O_{N,L}$ in \eqref{eq:outage_N_L}, 
	and $O_{N,N}$ in \eqref{eq:outage_N_N} for different values of fading shape factors of $m$ and $K$. It is seen that the values obtained through derived expressions (solid lines), exactly match the values obtained through Monte-Carlo simulation by generating corresponding Rician and Nakagami-m random values and taking the average of the outage. 
	Similar procedure have been used in order to validate the estimated frame decoding error probability expressions obtained in \eqref{eq:H_N} and \eqref{eq:H_L} as seen in Fig. \ref{fig:sim_epsilon_verify}. By considering 4-level piece-wise linearization ($L=4$), it is seen that both derived expressions of \eqref{eq:H_N} (first figure) and \eqref{eq:H_L} (second figure) tightly match the exact values.

	\subsubsection{Performance versus the height of UAV}
	
	\begin{figure*}
		\begin{minipage}{.45\linewidth}
			\centering
			\includegraphics [width=224pt,height=115pt]{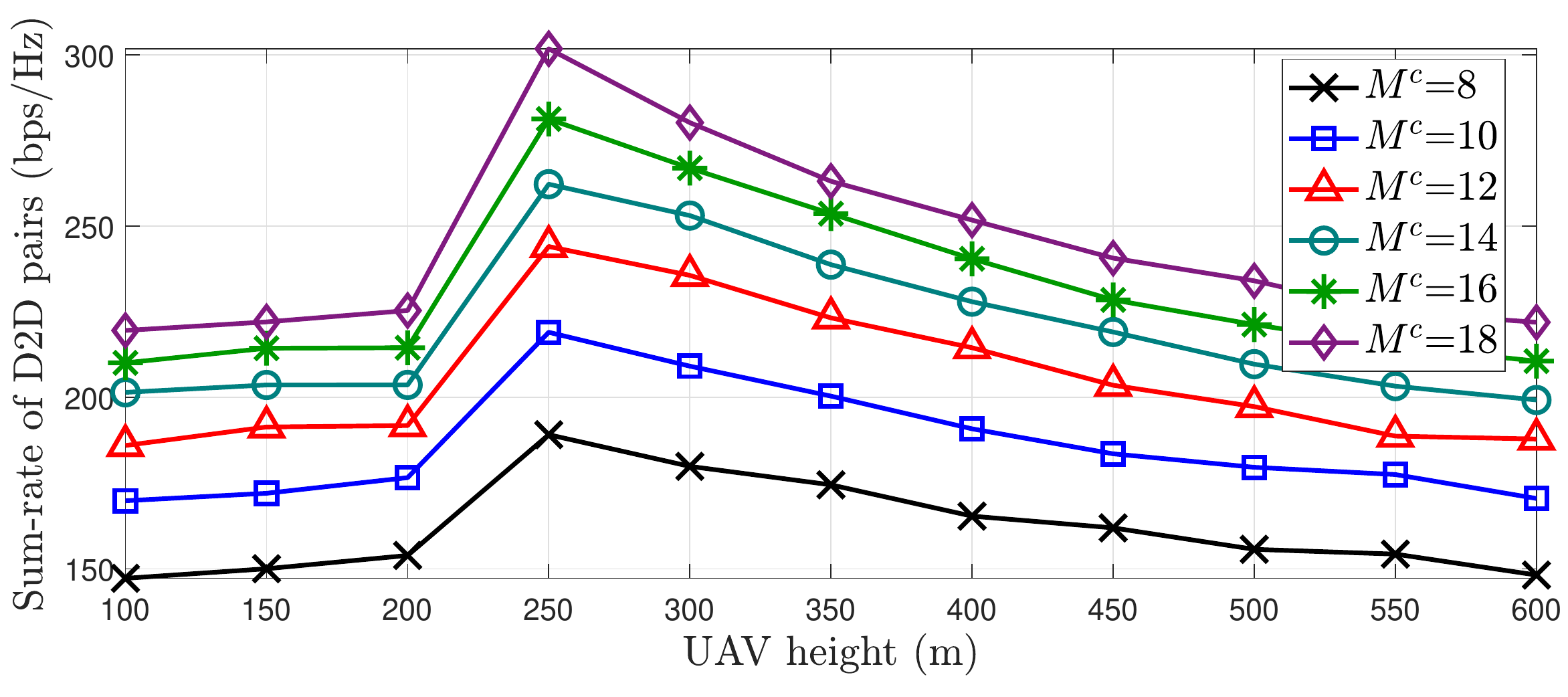}
			\includegraphics
			[width=224pt,height=105pt]{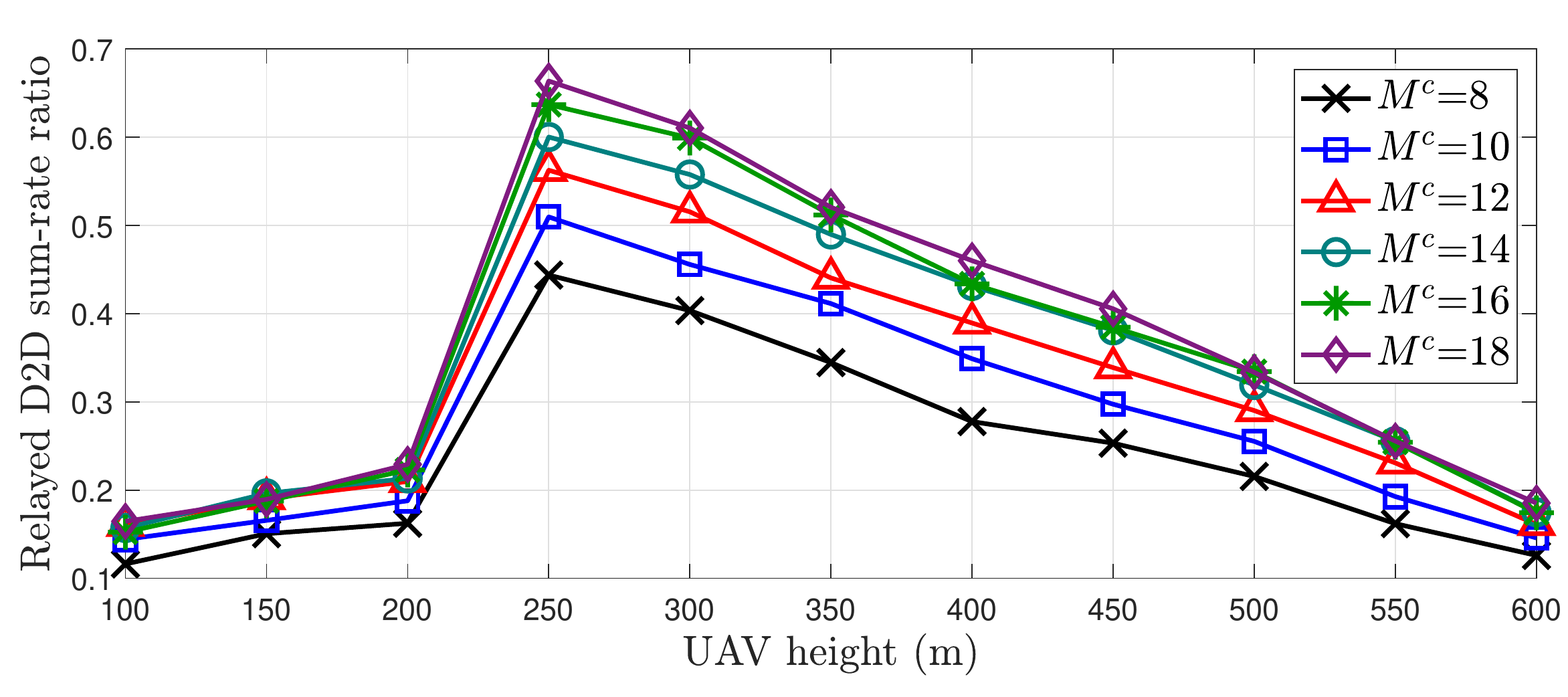} \\
			\caption{Comparison of the sum-rate and relayed D2D sum-rate ratio of D2D pairs versus the height of UAV and number of available cellular users.} 
			\label{fig:sim_h}
		\end{minipage}
		\hspace{20pt}
		\begin{minipage}{.45\linewidth}
			\centering
			\includegraphics [width=224pt,height=115pt]{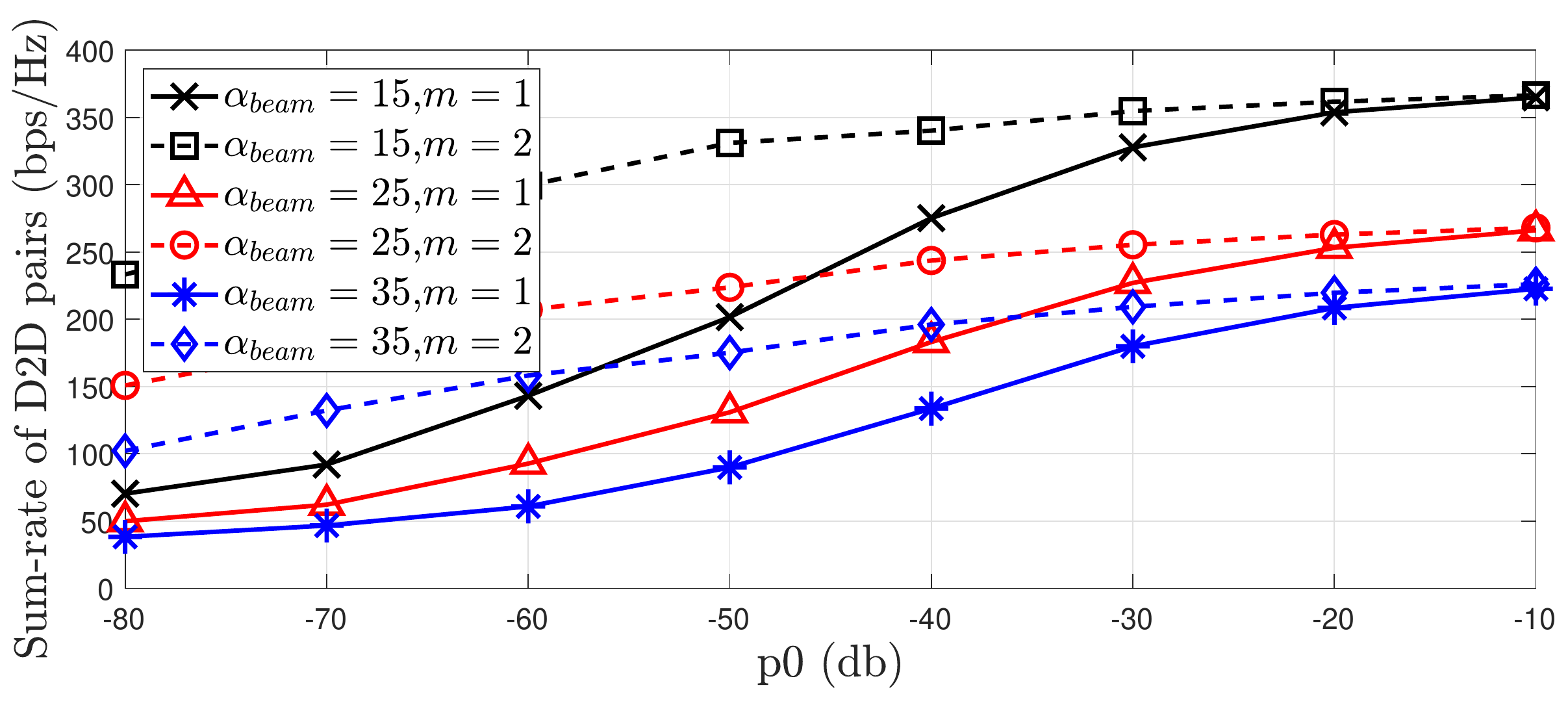} 
			\includegraphics
			[width=224pt,height=105pt]{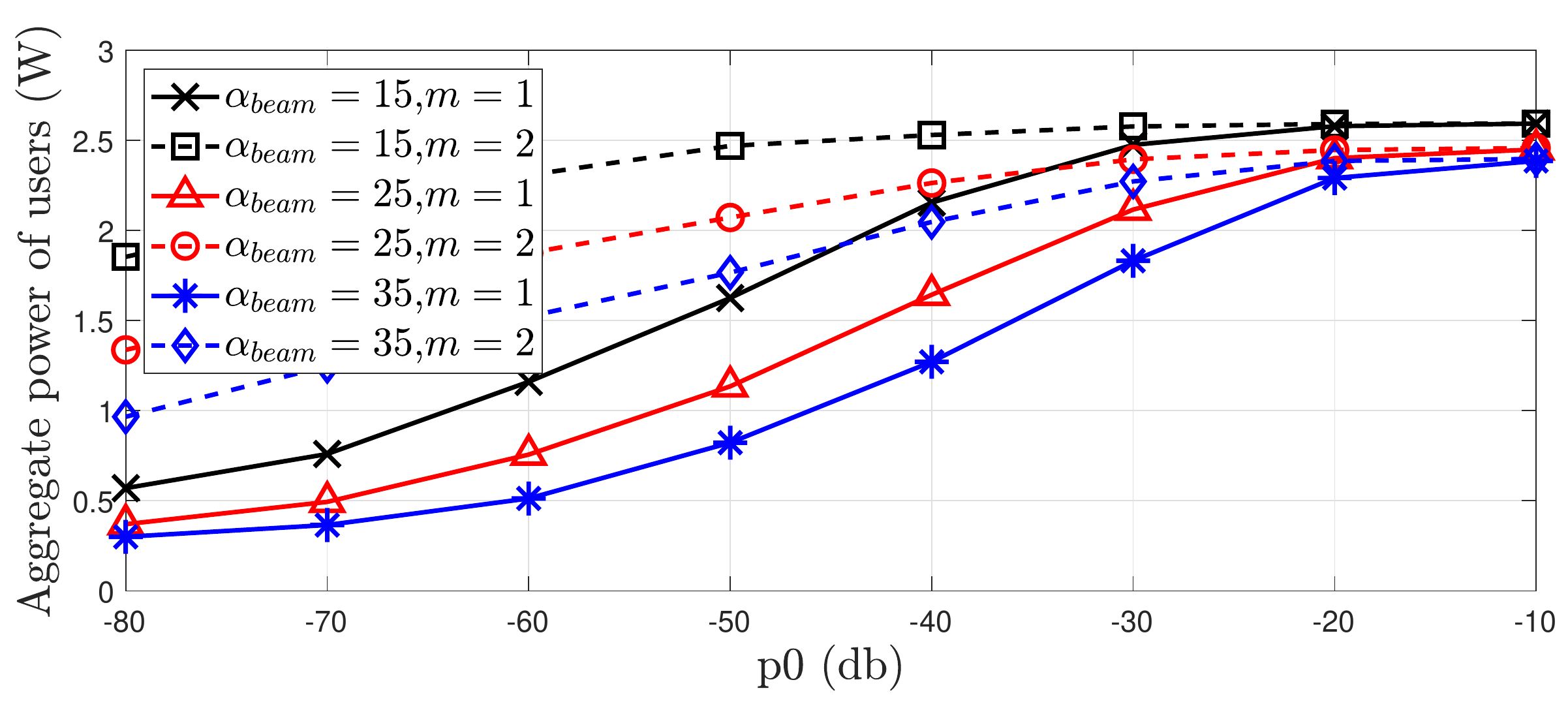} \\
			\caption{Comparison of the sum-rate of D2D pairs, and aggregate power of users versus the decoding error probability ($p_{\eps}$), and $\theta_{\textrm{3dB}}$ and fading shaping factor $m$.} 
			\label{fig:sim_p03}
		\end{minipage}
% 		\vspace{-10mm}
	\end{figure*}
	
	Fig. \ref{fig:sim_h} shows how increasing the height of UAV and also the number of cellular users influence the performance of the system by using Algorithm 2. We note that an increase in the number of cellular users provide more pairing options for D2D users which always leads to an increase in the system throughput. Also, we note that there is an optimal height (around $250$ m) beyond which the path loss due to high link distance, and below which the NLoS probability are the dominating factors of the performance degradation. For example, while it can be seen that for $M^c=16$, the relayed D2D sum-rate ratio (the ratio of the sum-rate of relayed D2D pairs to that of all D2D pairs) is $0.16$ and $0.17$ for UAV height of $100$ m and $600$ m respectively, it is near $0.64$ for the height of $250$ m corresponding to the optimal total sum-rate of $280$ bps/Hz for D2D pairs. {\em This figure also shows the significance of UAV relays compared to terrestrial relays.}

	\subsubsection{Performance versus decoding error probability }	
	Fig. \ref{fig:sim_p03} shows performance of Algorithm 2 versus maximum acceptable frame decoding error probability of cellular users (i.e., $p_\eps$) for the half-power angle $\theta_{3\textrm{dB}}$ values of $15^{\circ}$, $25^{\circ}$, and $35^{\circ}$ and fading shape factor of $m=1$ and $m=2$. Firstly, it is seen how increasing $p_\eps$ of cellular users increases the sum-rate of D2D pairs for different values of $m$ and $\theta_{3\textrm{dB}}$. This is because allowing cellular users to tolerate more frame decoding error probability will permit D2D users to  impose more interference on their underlaid cellular users. Secondly, it seen how increasing the fading shape factor $m$ (which corresponds to less fading variance) results in the increase of D2D sum-rate. The increase is more evident for lower values of $p_\eps$ since for lower values of $p_\eps$ more stringent power control is employed for D2D users with severe fading conditions of $m=1$. Finally, increasing the  half-power beamwidth results in decreasing the D2D sum-rate due to more interference imposed on cellular users.
	
	\subsubsection{UAV-assisted D2D vs Direct D2D}
	\label{sec:sim_uav}
	\begin{figure*}
		\begin{minipage}{.47\linewidth}
			\centering
			\includegraphics [width=230pt,height=115pt]{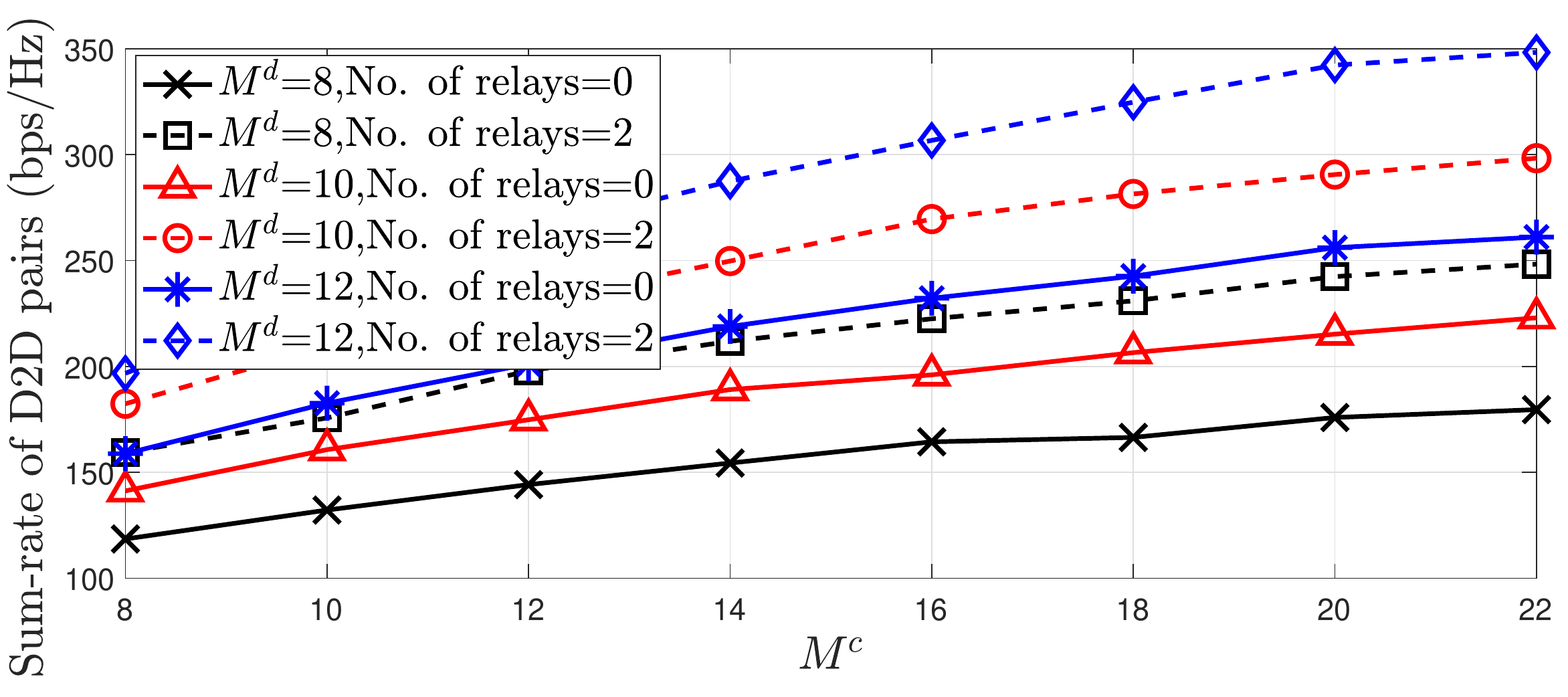} 
			%\includegraphics
			%[width=230pt,height=95pt]{pictures/simulation/-mc-/2_power.eps} \\
			\includegraphics [width=230pt,height=105pt]{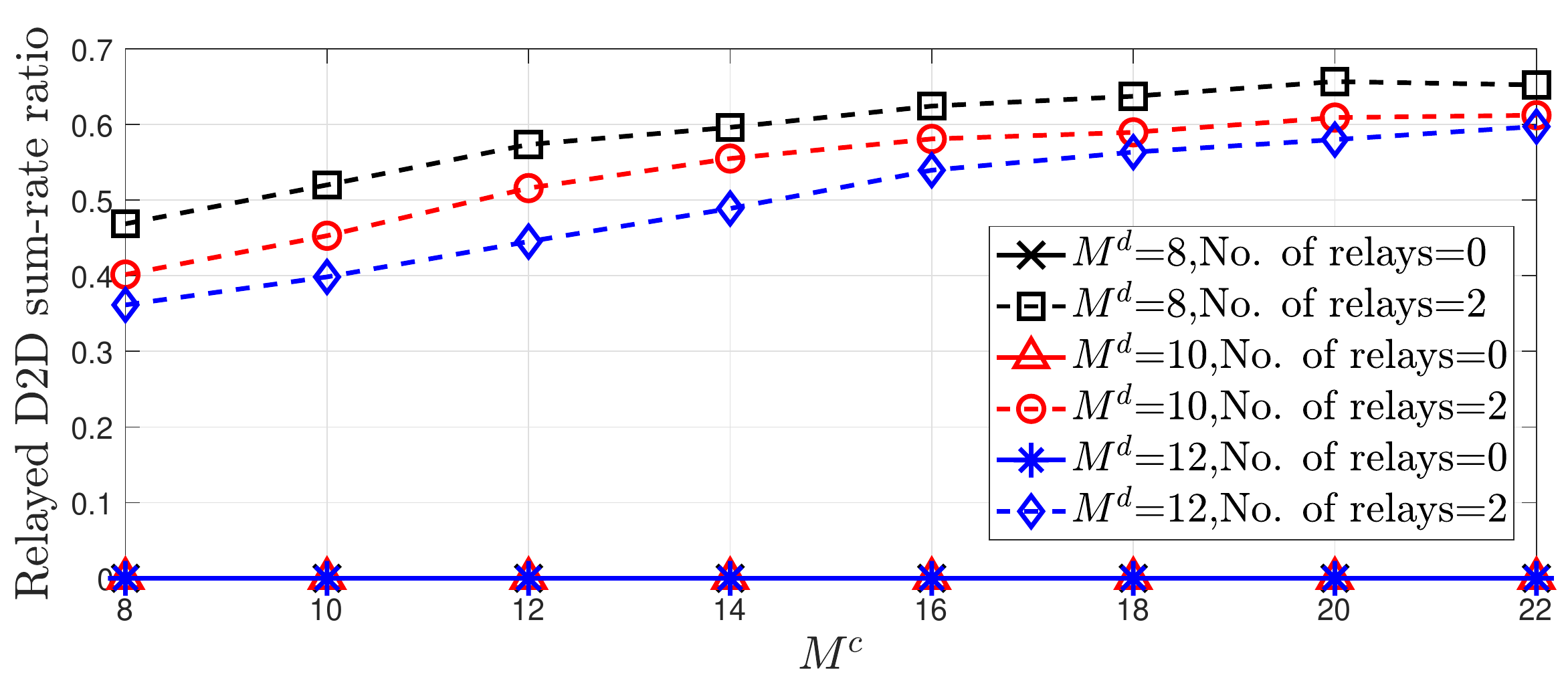} 
			\caption{Sum-rate of D2D pairs and relayed D2D sum-rate ratio for the cases when UAV is or is not employed versus number of cellular users and D2D pairs.} \vspace{-20pt}
			\label{fig:sim_mc}
		\end{minipage}
		\hspace{10pt}
		\begin{minipage}{.47\linewidth}
			\centering
			\includegraphics [width=230pt,height=115pt]{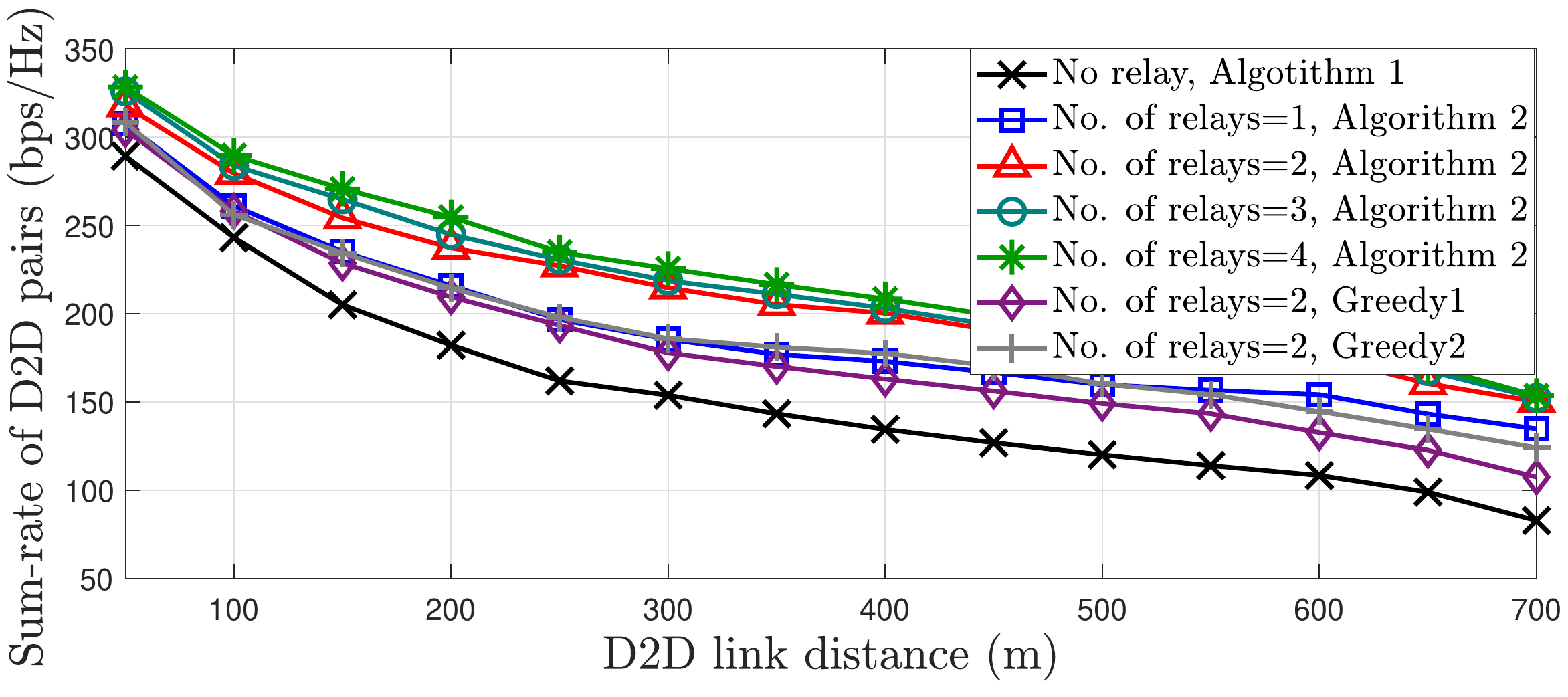} 
			\includegraphics
			%[width=230pt,height=105pt]{pictures/simulation/-d-/2_power.eps} 
			%\includegraphics
			[width=230pt,height=105pt]{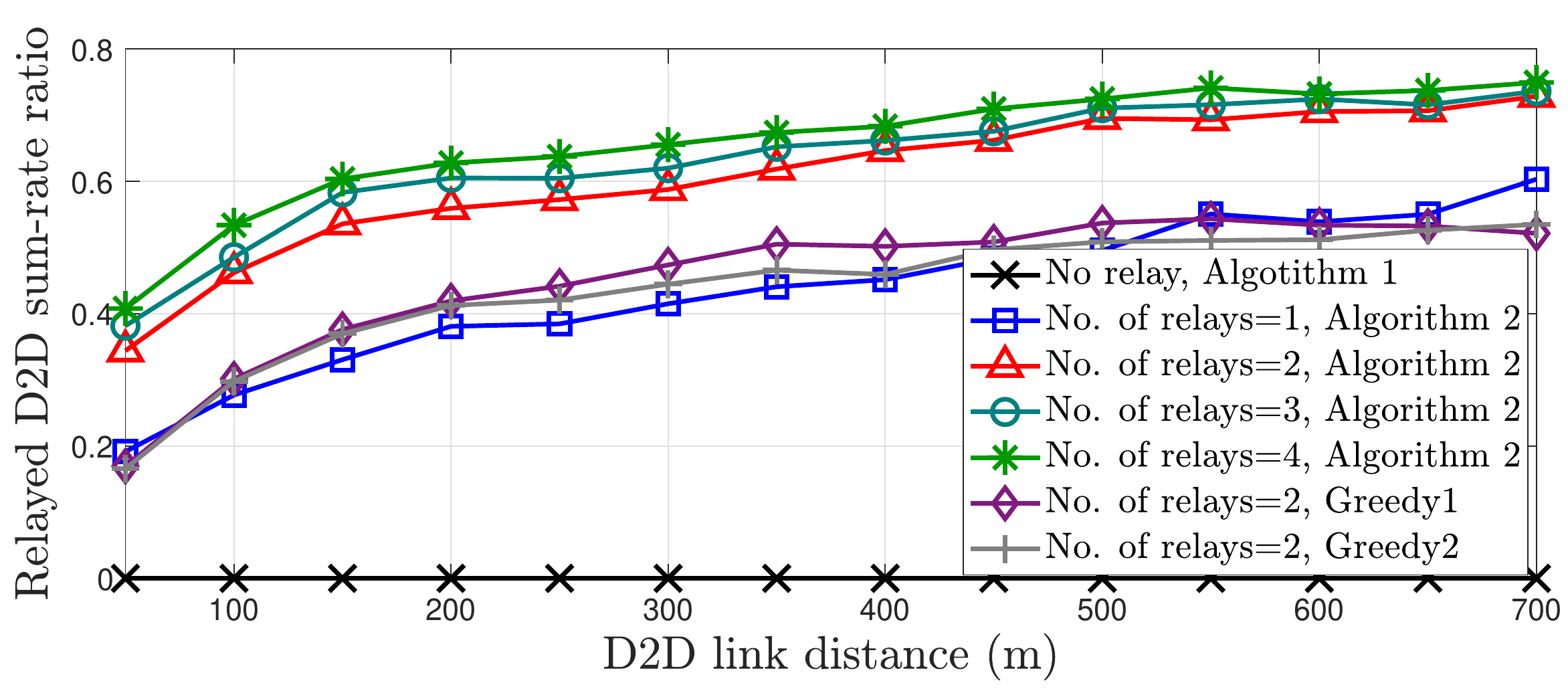} 
% 			\vspace{-40pt}
			\caption{Sum-rate of D2D pairs, aggregate power of users and relayed D2D sum-rate ratio versus the distance between D2D receiver and transmitter, and  number of UAV relays.} 
% 			\vspace{-20pt}
			\label{fig:sim_d}
		\end{minipage}
	\end{figure*}
	
	Fig. \ref{fig:sim_mc} illustrates the  sum-rate of D2D users, aggregate transmit power of users, and relayed D2D sum-rate ratio  versus the number of cellular users and number of D2D pairs computed through Algorithm~1 (with no relays) and Algorithm~2 (with UAV relays). It is seen how increasing $M^c$ and $M^d$ and also employing UAV relays can increase the aggregate throughput of D2D pairs. For example, for $M^c=14$ and $M^d=8$, Algorithm 2 enhances the D2D aggregate throughput from $154$ to $210$ bps/Hz which is about $36 \%$ enhancement. Also, increasing $M^d$ from $8$ to $12$, increases the throughput from $154$ to $219$ with no relay, and from $210$ to $285$ with UAV relays. %It is also seen that as $M^c$ increases, the performance enhancement of UAV-assisted scheme compared to non-UAV-assisted scheme is more increased. 

	To illustrate the significance of UAVs on the performance of D2D pairs with different link distances, we consider a maximum of 4 UAVs located at $(\pm 450, 0)$ and $(0, \pm 450)$. The  aggregate D2D throughput and relayed D2D sum-rate ratio versus the number of UAVs and D2D links' distance is shown in Fig.~\ref{fig:sim_d}. We have shown the performance of Algorithm 1 when no UAV exists, Algorithm 2 for 1-4  number of UAVs, and that of the proposed greedy algorithms when two UAVs are employed. As can be seen, Algorithm~2 outperforms the greedy algorithm to a great extent. We can also see that the increase in D2D link distance will generally reduce the aggregate throughput of D2D pairs, however, the higher link distance, the more UAVs assist to the sum-rate of D2D pairs. For example, it is seen that for the link distance of $200$ m, Algorithm 1 (with no UAV relay) obtains the sum-rate of $182$ bps/Hz, while Algorithm 2 with 1 and 2 UAVs obtains the sum-rate of $214$ and $237$ bps/Hz which correspond to $17\%$ and $30\%$ enhancement, respectively. However, for D2D link distance of $600$ m, Algorithm 1 (with no relay) obtains the sum-rate of $108$ bps/Hz, while Algorithm 2 with 1 and 2 UAVs obtains the sum-rate of $154$ and $176$ bps/Hz which correspond to $43\%$ and $63\%$ enhancement respectively. Besides, it is seen that employing more than two UAVs will not have significant impact on the performance. 
	% 	Another interesting point seen in Fig. \ref{fig:sim_d} is that while the increase in D2D link distance increases the aggregate power for Algorithm 2 with UAV relays due to more power needed to compensate for the path-loss, however this is reverse in Algorithm 1 with no UAV relay. This is justified by noting that the increase in D2D link distance decreases $\eta^*$ mentioned in  Theorem \ref{th:power} which in turn forces more D2D pairs to choose zero power level as seen in  \eqref{eq:opt_reused_subproblem_power_pair}.
	\textcolor{black}{
		\subsubsection{Exploring the Effect of Intercell Interference}
		\label{sec:sim_multicell}
		\begin{figure}
			\begin{minipage}{.5\linewidth}
				% 			\centering
				\includegraphics [width=230pt,height=115pt]{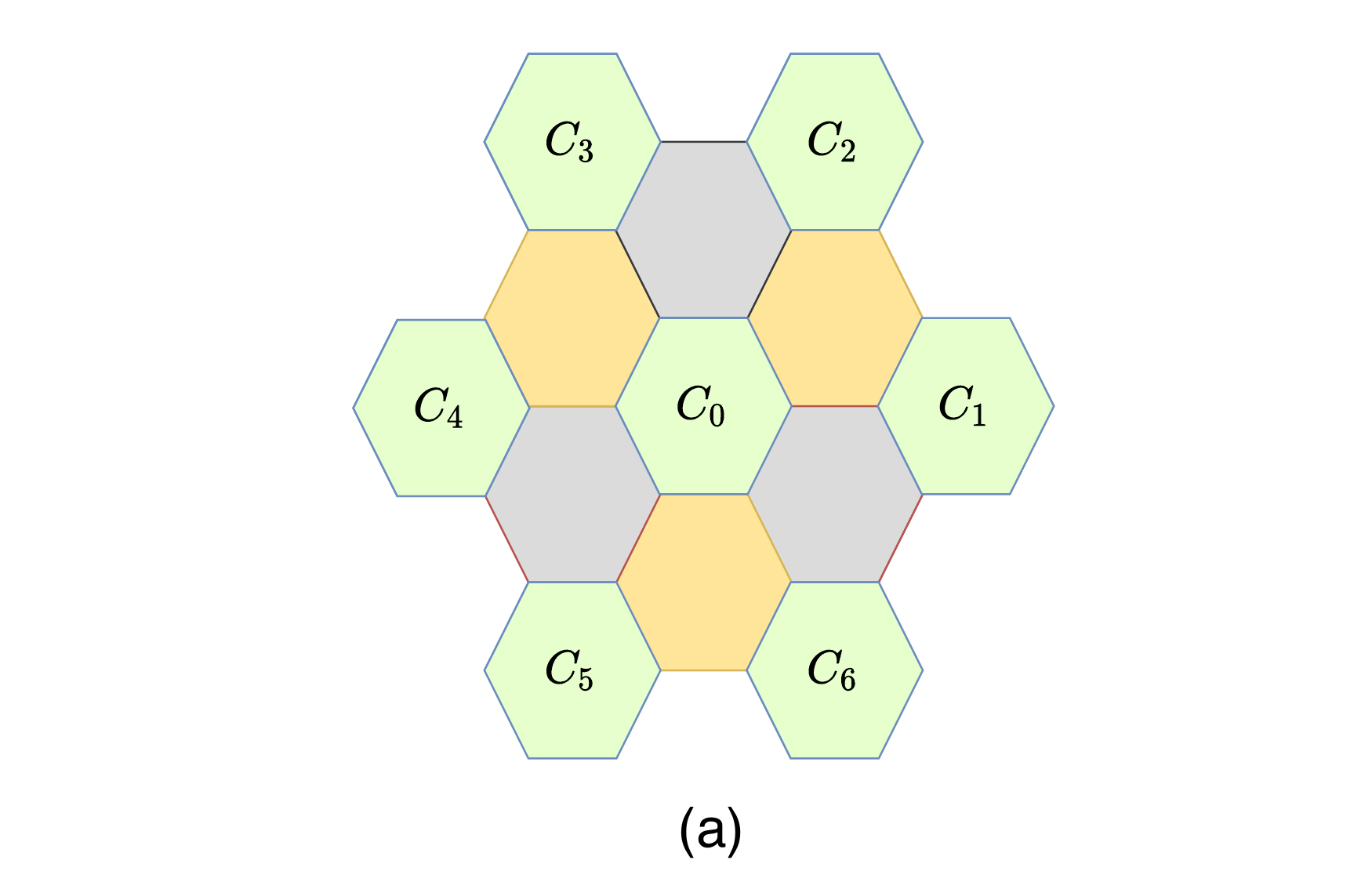} 
				\label{fig:sim_multicell}
			\end{minipage}
			\hfill
			\begin{minipage}{.5\linewidth}
				% 			\centering
				\includegraphics [width=230pt,height=125pt]{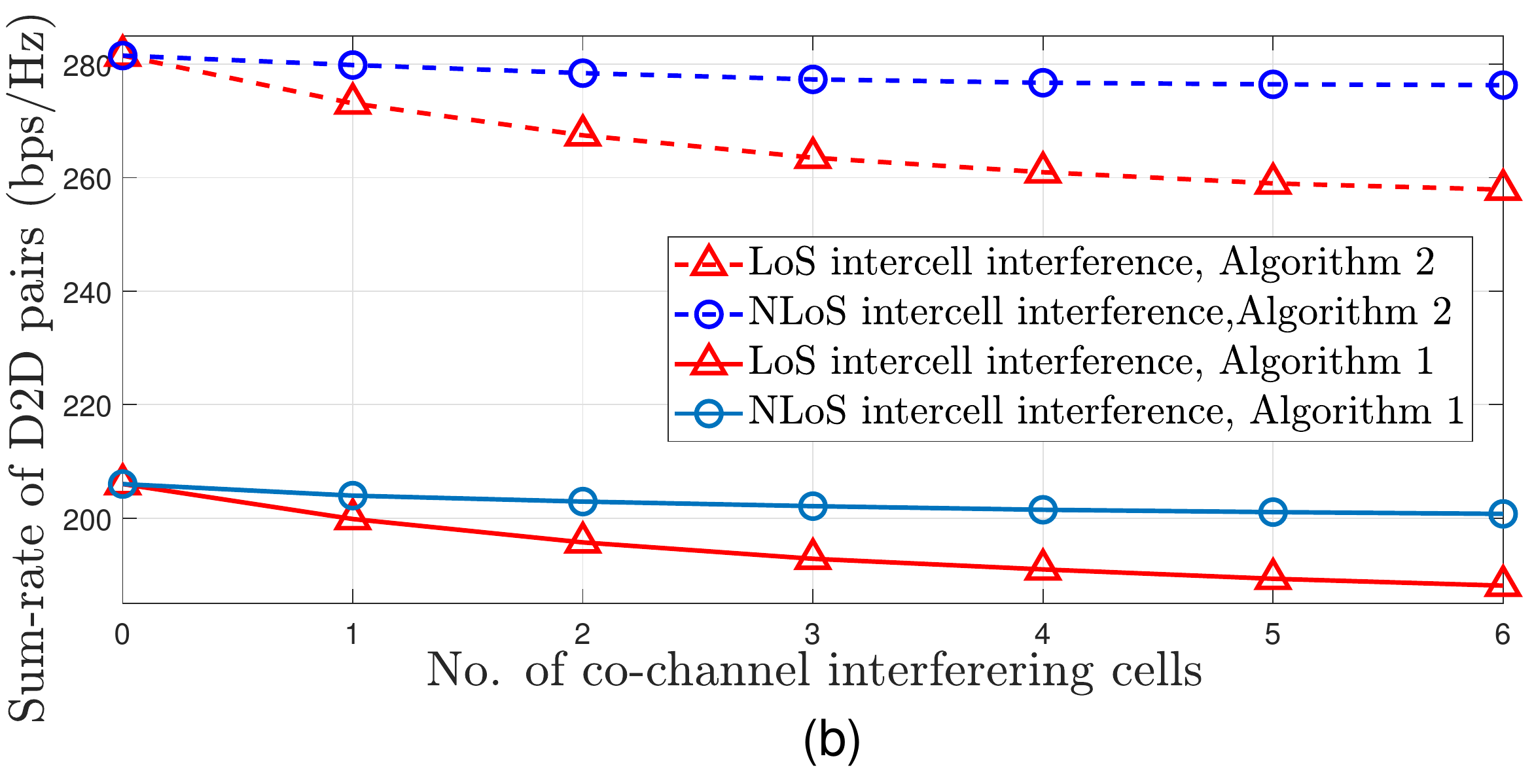}
			\end{minipage}
%			\hfill
%			\begin{minipage}{.33\linewidth}
%				% 		\centering
%				\includegraphics [width=170pt,height=130pt]{pictures/simulation/-mcopt-/1_D2D_rate.eps} 
%				% 		 \vspace{-20pt}
%				% 		\caption{Investigating the optimality of Algorithm 2.} 
%				% 		\vspace{-20pt}
%				\label{fig:sim_mcopt}
%			\end{minipage}
% 			\vspace{-20pt}
			\caption{\textcolor{black}{Investigating the effect of intercell co-channel interference. Fig. (a) depicts the network structure, wherein central green cell serves D2D pairs, and cellular users from surrounding green cells impose co-channel intercell interference. Fig. (b) illustrates the performance of Algorithm 1 and Algorithm 2 for different number of  co-channel interfering cells.}}
		\end{figure}
		To investigate the effect of co-channel intercell interference, we consider a multi-cell cellular network according to Fig.~\ref{fig:sim_multicell}-(a), wherein the central cell $C_0$ serves D2D pairs according to Algorithms 1 and 2, and a subset of surrounding co-channel interfering cells $C_1$-$C_6$ impose interference due to frequency reuse. Fig. \ref{fig:sim_multicell}-(b) shows the performance  versus different number of co-channel interfering cells, considering LoS and NLoS intercell interference\footnote{By using the terms LoS and NLoS interferences, we mean that the channel between the interfering transmitter and the desired receiver is LoS and NLoS, respectively.}. It is observed that for both algorithms,  NLoS intercell interference has little impact on the performance compared to LoS intercell interference. This is due to high NLoS path-loss exponent and  high distance of interfering source from non-adjacent co-channel cells.
% 		However, for LoS interference, the increase in the number of interfering cells is seen to have rather considerable impact on the performance degradation of the system.
	}

	\subsubsection{Optimality of Algorithm 2}
	 	\begin{figure}
	 		%\begin{minipage}{.47\linewidth}
	 		\centering
	 		\includegraphics [width=230pt,height=130pt]{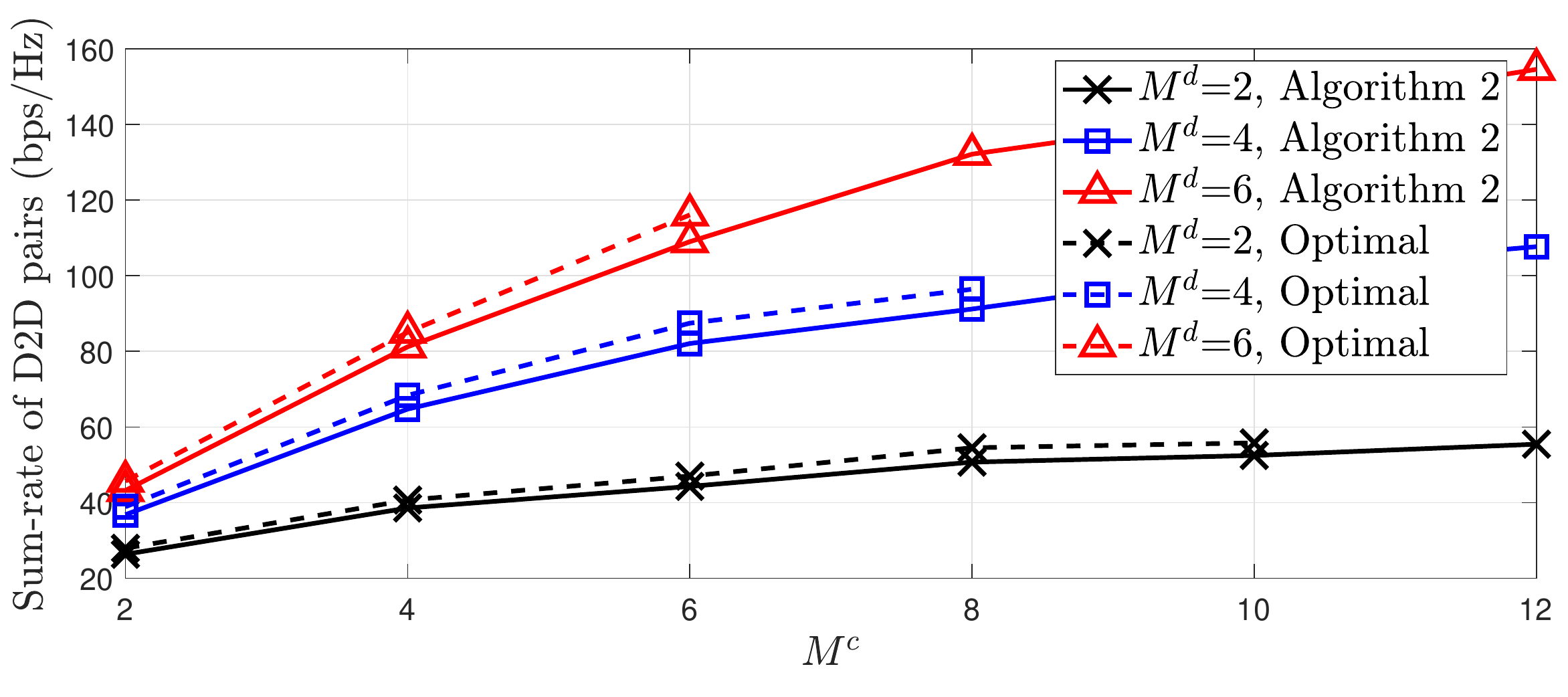} 
	 		%  \vspace{-20pt}
	 		\caption{Investigating the optimality of Algorithm 2.} 
	 		% \vspace{-20pt}
	 		\label{fig:sim_mcopt}
	 		%\end{minipage}
	 	\end{figure}
	Finally, we study the optimality of Algorithm~2 compared to globally optimal solution obtained by exhaustive search through all possible channel and link-type options. 
	Fig. \ref{fig:sim_mcopt} illustrates this comparison versus different total number of users.
	%	Fig . *** compares the globally optimal solution of Problem *** obtained through exhaustive search, to that of proposed Algorithm 2 for different number users.
	Note that due to the exponential complexity of the exhaustive search, we are only able to present numerical results for small total number of D2D and cellular users. It is seen that the performance of our proposed algorithm is near-optimal for the evaluated total number of cellular and D2D users.

	\section{Conclusion}	
	\label{sec:conclusions}
	In this paper, we derived analytical closed form expressions of outage probability, ergodic capacity, and frame decoding error probability  taking into account Nakagami-m and Rician fading channel models for terrestrial and aerial transmissions, respectively. The expressions were employed to devise efficient and ultra-reliable resource (subchannel, power, and link-type) management of D2D pairs and cellular users. Simulation results demonstrate the significance of aerial relays compared to ground relays in increasing the throughput of D2D pairs especially for distant D2D pairs. The framework can be extended to optimize the trajectory of UAV relays and  a large scale network with multiple interferers. The characterization of frame decoding error probability can be extended  under non-finite and non-quasi-static block regimes, and also, the proposed resource allocation can be extended for multi-hop UAV relaying in future works.

	\begin{appendices}
		\renewcommand{\thesectiondis}[2]{\Alph{section}:}
		\section{Proof of Lemma \ref{lm:outage_L_N}}
		\label{apx:outage_L_N}
		Let $y$ and $x$ be unit mean random variables corresponding to the fading channels of the main LoS link $k$ and the interfering NLoS link $k'$ respectively. From (\ref{eq:gamma_k}) we obtain \eqref{eq:apxa} expressed at the top of  next page 
		\begin{figure*}
		\begin{align*}
		\label{eq:apxa}
		&O_{L,N}({\alpha}) 
		= \textrm{Pr}\{ y < \alpha x  \} 
		=  \int_{0}^{\infty} \int_{0}^{\alpha x} f_G(x,m) f_R(y,K) dy dx
		\\
		&= \KK \times
		\int_{0}^{\infty} \int_{0}^{\alpha x} x^{m-1} e^{-mx-\KP y} I_0(\sqrt{4K\KP y}) dydx \nonumber \\
		%& = \KK \int_{0}^{\infty} x^{m-1} e^{-mx} \int_{0}^{\alpha x} e^{-\KP y} I_0(\sqrt{4K\KP y})dydx \nonumber \\
		& \stackrel{(a)}{=} \KK \times
		\int_{0}^{\infty} x^{m-1} e^{-mx} \int_{0}^{\alpha x} e^{-\KP y} \sum_{j=0}^{\infty} \frac{(K \KP)^j y^j}{(j!)^2} dydx \nonumber \\
		& = \KK  \times
		\sum_{j=0}^{\infty} \frac{(K \KP)^j}{(j!)^2} \int_{0}^{\infty} x^{m-1} e^{-mx} \int_{0}^{\alpha x} e^{-\KP y}   y^j dydx \nonumber \\
		& \stackrel{(b)}{=} \KK \sum_{j=0}^{\infty} \frac{(K \KP)^j}{(j!)^2} \times
		\int_{0}^{\infty} x^{m-1} e^{-mx}  \left[-\frac{e^{-\KP y}j!}{\KP^{j+1}}\sum_{k=0}^j \frac{(\KP y)^k}{k!} \right]_0^{\alpha x} dx \nonumber \\
		%& = \frac{m^{m}}{\Gamma(m)e^{K}} \sum_{j=0}^{\infty} \frac{(K)^j}{j!} \times
		%\int_{0}^{\infty} x^{m-1} e^{-mx} \left( 1- e^{-\alpha \KP x}\sum_{k=0}^j \frac{(\alpha \KP x)^k}{k! }\right)  dx\nonumber \\
		& \stackrel{(c)}{=} \frac{m^{m}}{\Gamma(m)e^{K}} \sum_{j=0}^{\infty}  \frac{(K)^j}{j!} \times
		\left( \frac{\Gamma(m)}{m^{m}}- \sum_{k=0}^j \frac{(\alpha \KP)^k}{k!} \int_{0}^{\infty} e^{-x(m+\alpha \KP)}  x^{k+m-1}  dx \right) \nonumber \\
		& \stackrel{(d)}{=} 1 - \frac{m^m e^{-K} }{\Gamma(m)} \sum_{j=0}^{\infty} \frac{K^j}{j!} \sum_{k=0}^j \frac{(\alpha \KP)^k \Gamma(k+m)}{k!(m+\alpha \KP)^{k+m}}    = 1 - \frac{\frac{m^m}{\Gamma(m)} e^{-K}}{(m+\alpha \KP)^{m}}
		\sum_{j=0}^{\infty} \frac{K^j}{j!} \sum_{k=0}^j  \left(\frac{\alpha \KP}{m+\alpha \KP}\right)^k \! \frac{\Gamma(k+m)}{k!}
		\end{align*}
%		\hline
		\end{figure*}
		wherein (a) is derived by expanding the modified Bessel function as $I_0(x)=\sum_{k=0}^{\infty} \frac{(x/2)^{2k}}{(k!)^2}$ (\cite{gradshteyn2014table}, Eq. 8.447), (b) and (d) are derived using the equality $\int e^{-cx} x^{m-1} dx = -e^{-cx}\frac{\Gamma(m)}{c^m}\sum_{k=0}^{m-1} \frac{(cx)^k}{k!}$, and (c) holds due to the equality $\int_{0}^{\infty} e^{-c x} x^{m-1} dx=\frac{\Gamma(m)}{c^{m}}$.
		Letting $\theta=\frac{\alpha \KP}{m+\alpha \KP}$ results in
		\begin{align*}
		O_{L,N}(\alpha)  
		 = 1 - \frac{  \frac{{m}^{m}}{\Gamma({m})} e^{-K} \sum_{j=0}^{\infty} \frac{K^j}{j!} \frac{d^{{m}-1}}{d \theta^{m-1}} \left[ \sum_{k=0}^j  \theta^{k+{m}-1} \right]}		{({m}+\alpha \KP)^{m}}  \\
		%& = 1 - \frac{ \frac{{m}^{m}}{\Gamma({m})} e^{-K}}{({m}\!+\!\alpha \KP)^{m}} \frac{d^{m-1}}{d \theta^{m-1}} \! \left[ \theta^{m-1} \sum_{j=0}^{\infty} \frac{K^j}{j!} \! \times \! \frac{1-\theta^{j+1}}{1-\theta} \right] \\
		 = 1 - \frac{ {m}^{m} e^{-K}}{\Gamma({m})({m}+\alpha \KP)^{{m}}} \frac{d^{m-1}}{d \theta^{{m}-1}} \left[ \frac{\theta^{{m}-1}}{1-\theta} ( e^K - \theta e^{K\theta} )  \right]. \numberthis
		\end{align*}
		\section{Proof of Lemma \ref{lm:outage_N_N}}
		\label{apx:outage_N_N}
		Let $y$ and $x$ be unit mean Gamma distributed random variables with shape factors $m$ and $m'$ corresponding to the desired link $k$ and interfering link $k'$ respectively. We have
		\begin{align*}
		& O_{N,N}(\alpha)= \textrm{Pr}\{ y < \alpha x  \}
		=\int_{0}^{\infty} \! \int_{0}^{\alpha x} \!\! f_G(x,m') f_G(y,m) dy dx 
		% 	 	\\
		% 	 	& = \frac{m^{m} m'^{m'}}{\Gamma(m) \Gamma(m')}  \int_{0}^{\infty} \int_{0}^{\alpha x} x^{m'-1}  e^{-m' x} y^{m-1} e^{-m y}  dydx 
		\\
		& = \frac{m^{m} m'^{m'}}{\Gamma(m) \Gamma(m')} \int_{0}^{\infty} x^{m'-1}  e^{-m' x} \int_{0}^{\alpha x} y^{m-1} e^{-m y}  dydx 
		\\
		%& = \frac{m^{m} m_2^{m_2}}{\Gamma(m) \Gamma(m_2)}   \int_{0}^{\infty} x^{m -1} e^{-m x} \times \\
		%		& \hspace{80pt} \left[-e^{-m_2 y}\Gamma(m_2)\sum_{k=0}^{m_2-1} \frac{y^k}{k! m_2^{m_2-k}} \right]_0^{\alpha x}  dx \\		
% 		& \stackrel{(a)}{=} \frac{m^{m} m'^{m'}}{\Gamma(m) \Gamma(m')}   \int_{0}^{\infty} x^{m' -1} e^{-m' x} \times 
% 		\left(\frac{\Gamma(m)}{{m}^{m}}-e^{-m \alpha x} \frac{\Gamma(m)}{m^{m}} \sum_{k=0}^{m-1} \frac{(\alpha x)^k}{k!} \right) dx 
		\\		
		& \stackrel{(a)}{=}  1 - \frac{m'^{m'}}{\Gamma(m')} \int_0^\infty x^{m' -1}e^{-x(m' +m \alpha) }  \sum_{k=0}^{m-1}  \frac{( m \alpha x)^k}{k!}  dx  
	 	 	\\&		
		% 	 	& = 1 - \frac{m'^{m'}}{\Gamma(m')} \sum_{k=0}^{m-1} \frac{(m \alpha)^k}{k!} \int_0^\infty x^{k+m' -1}e^{-x(m' +m \alpha) }  dx  
		\stackrel{(b)}{=} 1 - \frac{m'^{m'}}{\Gamma(m')} \sum_{k=0}^{m-1} \frac{(m \alpha)^k \Gamma(k+m')}{k! (m' +m \alpha)^{k+m'}},	
		\end{align*}
		where (a) and (b) are derived using the equalities $\int e^{-cx} x^{m-1} dx = -e^{-cx}\frac{\Gamma(m)}{c^m}\sum_{k=0}^{m-1} \frac{(cx)^k}{k!}$ and $\int_{0}^{\infty} e^{-c x} x^{m-1} dx=\frac{\Gamma(m)}{c^{m}}$, respectively.
		\section{Proof of Theorem \ref{th:decoding}}
		\label{apx:decoding}
		We can verify that $u(\gamma)$ is monotonically increasing and hence $Q(u(\gamma))$ is a monotonically decreasing function of $\gamma$. We can also show that $\underset{\gamma \rightarrow 0}{\mathrm{lim\ }} \eps_n(\gamma) = Q(0)=0.5$, and $\underset{\gamma \rightarrow \infty}{\mathrm{lim\ }} \eps_n(\gamma) = 0$. Fig. \ref{fig:q} shows a 2-level piece-wise linear approximation of \eqref{eq:Qu}. To achieve a better approximation, we can extend that to $L$-level piece-wise linear approximation, and estimate $\eps_n(\gamma)$ as:
		%\begin{equation}
		\begin{align}
		\eps_n(\gamma) \approx 
		\begin{cases}
		-\omega_i \gamma + \omega_i \gamma_i + 0.5(1-\frac{i}{L}), 
		& \textrm{if }  \gamma_{i-1}\leq\gamma < \gamma_i,
		\nonumber
		\\
		0, 
		& \textrm{if } \gamma\geq \gamma_L,
		\nonumber
		\end{cases}
		\end{align}
		%\end{equation}
		where $\omega_i=(\eps_n(\gamma_{i-1})-\eps_n(\gamma_i))/(\gamma_i - \gamma_{i-1})=\frac{0.5/L}{\gamma_i - \gamma_{i-1}}$, $\gamma_i=\eps^{-1}(0.5(1-i/L))$ for $0\leq i \leq L-1$, and $\gamma_L=\eps_n^{-1}(\Delta)$, in which $\Delta \ll 0.5/L$ is a small value constant. 
		%      \begin{figure}[!t]
		% 		\centering
		% 		\includegraphics [width=254pt,height=110pt]{pictures/extra/q.pdf} 
		% 		\vspace{-25pt}
		% 		\caption{Two-level piece-wise approximation of $\eps_n$.}
		% 		\label{fig:q}%\vspace{-10mm}
		% 	\end{figure}

		\begin{figure}
			\centering
			\includegraphics [width=184pt,height=100pt]{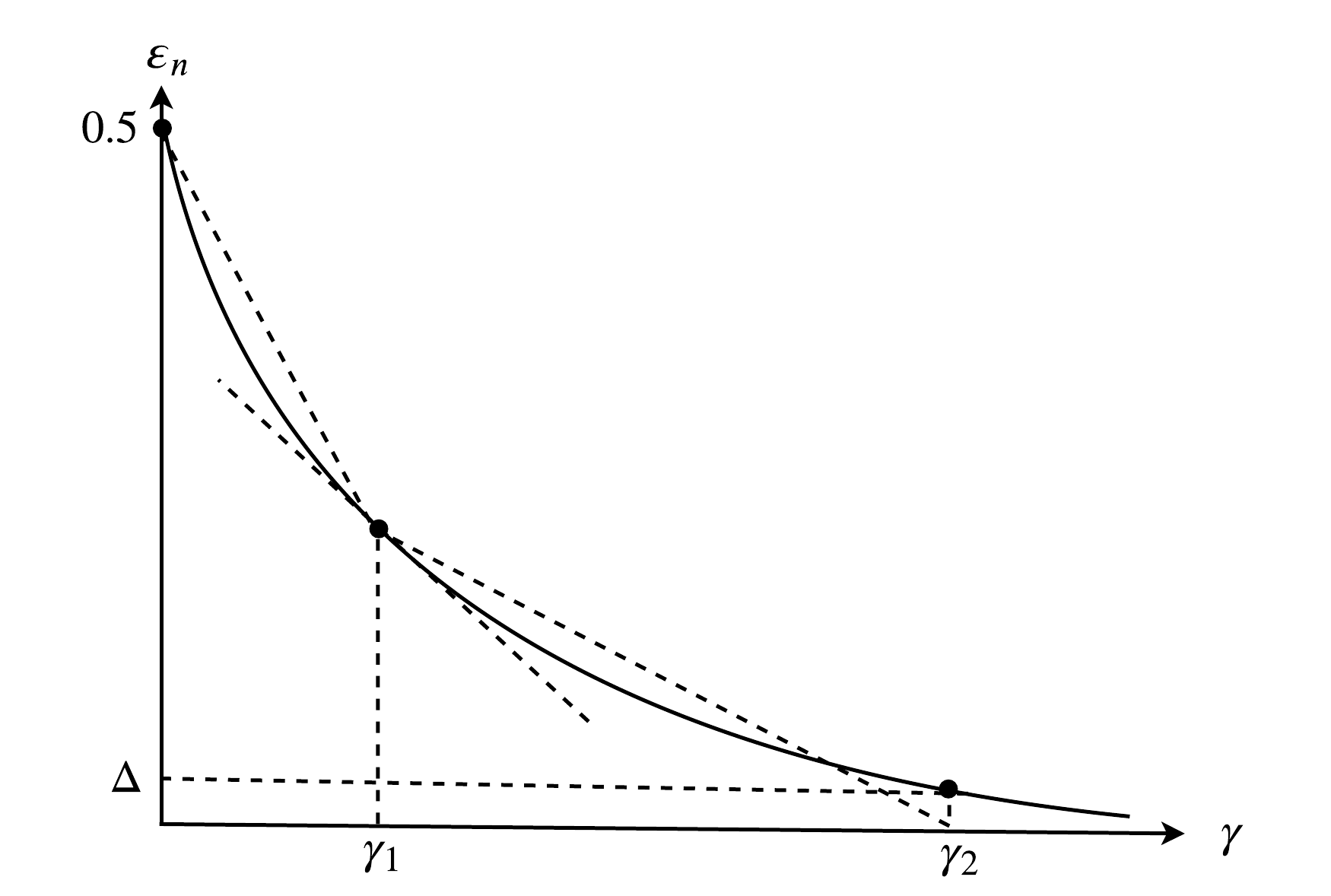} 
% 			\vspace{-20pt}
			\caption{Two-level piece-wise approximation of $\eps_n$.}
			\label{fig:q}
% 			\vspace{-10mm}
		\end{figure}

		The average decoding error probability is then obtained as follows:
		\small
		\begin{align}
		&\overline{\eps}_n  = \mathbb{E}\{\eps_n(\gamma)\}  = \int_0^\infty f_\gamma(\gamma) \eps_n(\gamma) d\gamma=
		K_0 +\sum_{i=1}^L \omega_i\int_{\gamma_{i-1}}^{\gamma_i} F_{\gamma}(\gamma) d \gamma
		\notag\\&=
		K_0 + \sum_{i=1}^{L-1}(\omega_{i-1}-\omega_i)H(\zeta_{j,i},\gamma_i) + \omega_L H(\zeta_{j,i},\gamma_L),
		\end{align}\normalsize
		where 
		\begin{equation}
		\label{eq:H}
		H(\alpha,\gamma)=\int_0^\gamma F_\gamma(\gamma) d\gamma = \int_0^\gamma O^c(\alpha \gamma) d \gamma,
		\end{equation}
		and $K_0$ is obtained as follows:
		\begin{align*}
		&K_0 = \sum_{i=1}^L -\omega_i \left(\gamma_i F(\gamma_i) -\gamma_{i-1}F(\gamma_{i-1})\right)
		% \\
		% &\hspace{40pt}
		\\
		&\hspace{55pt}
		+\left(\omega_i\gamma_i +0.5(1-i/L)\right)(F(\gamma_i)-F(\gamma_{i-1}))
		% \\
		% & = \sum_{i=1}^L 
		%  F(\gamma_{i-1})\omega_i (\gamma_{i-1}-\gamma_i) +0.5(1-\frac{i}{L})(F(\gamma_i)-F(\gamma_{i-1}))
		\\
		& \stackrel{(a)}{=} \sum_{i=1}^L 
		\frac{0.5}{L} F(\gamma_{i-1}) +0.5(1-{i}/{L})(F(\gamma_i)-F(\gamma_{i-1}))
		\\
		& = 0.5\sum_{i=1}^L 
		F(\gamma_i)\left(1-{i}/{L}\right) - F(\gamma_{i-1})\left(1-({i-1})/{L}\right)
		\\&= 0.5\left(F(\gamma_0)+F(\gamma_L)(1-{L}/{L})\right)=0,
		\end{align*} \normalsize
		where (a) follows from $\omega_i=\frac{0.5/L}{\gamma_i - \gamma_{i-1}}$. 
		For the NLoS interference, from \eqref{eq:H}, \eqref{eq:outage_N_N}, and \eqref{eq:O_c_N} we calculate $H(\alpha,\gamma)$ for the cellular user as follows:
		\begin{align*}
		H_N(\alpha&,\gamma)=\int_0^\gamma O^c(\alpha \gamma) d \gamma 
		= \int_0^\gamma O_{N,N} (\alpha \gamma) d \gamma 
		\\
		& =  \gamma - 
		\frac{m'^{m'}}{\Gamma(m')} \sum_{k=0}^{m-1} \int_0^\gamma \frac{(m \alpha \gamma)^k \Gamma(k+m')}{k! (m' +m \alpha \gamma)^{k+m'}}d \gamma
		\\
		& = \gamma - 
		\frac{m'^{m'}}{\Gamma(m')} \sum_{k=0}^{m-1} 
		\frac{\Gamma(k+m')(m\alpha)^k \gamma^{k+1}}{(k+1)! (m')^{k+m'}} 
		\\
		& \hspace{20pt}
		\times {}_2 F_1 \left(k+1,k+m',k+2,-\frac{m\alpha\gamma}{m'}\right).
		\numberthis
		\end{align*}
		Similarly, for the LoS interference scenario, by using from \eqref{eq:H}, \eqref{eq:outage_N_L} and \eqref{eq:O_c_L} we have
		\begin{align*}
		&H_L(\alpha,\gamma)=\int_0^\gamma O^c(\alpha \gamma) d \gamma 
		= \int_0^\gamma O_{N,L}(\alpha \gamma) d \gamma
		\\
		&
		= \int_0^\gamma \left(1- O_{L,N}({1}/{\alpha \gamma})\right) d\gamma
		\\
		& =  
		\frac{1}{\Gamma(m)} e^{-K}\sum_{j=0}^{\infty} \frac{K^j}{j!} 
		\sum_{k=0}^j 
		\frac{\Gamma(k+m)\KP^k }{k!} 
		% \\
		% &\hspace{90pt}
		\\
		& \hspace{115pt}
		\times \int_0^\gamma \frac{ (m\alpha \gamma)^m }{\left(K+1+m \alpha \gamma \right)^{k+m}} d\gamma
		\\
		&=
		\frac{ e^{-K} (m\alpha )^m  \gamma^{m+1}} {\Gamma(m)  (K+1)^m (m+1)}\sum_{j=0}^{\infty} \frac{K^j}{j!} 
		\sum_{k=0}^j 
		\frac{\Gamma(k+m) }{k!} 
		% \\
		% &\hspace{20pt}
		\\
		& \hspace{70pt}
		\times 
		{}_2 F_1 \left(m+1,m+k,m+2,-\frac{m\alpha\gamma}{K+1}\right).
		% 	  \numberthis
		\end{align*}
% 		\vspace{-5mm}
		%This completes the proof.
	\end{appendices}

% 	\footnotesize
	\bibliographystyle{IEEEtran}
	\bibliography{Mybib}
	
\begin{IEEEbiography}[{\includegraphics[width=1in,height
=1.25in,clip,keepaspectratio]{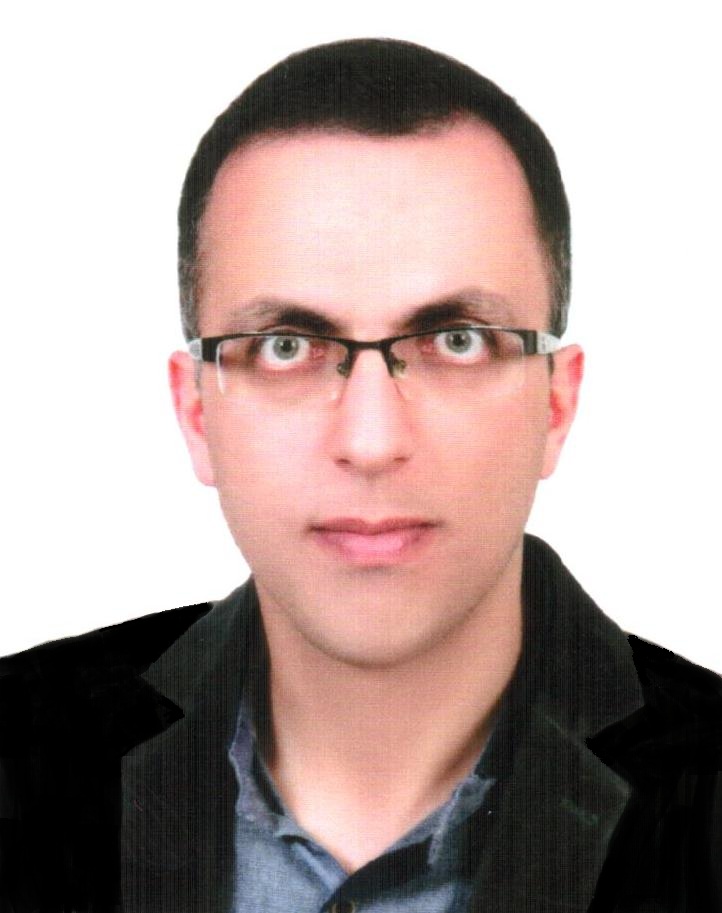}}]{Mehdi Monemi}
Mehdi Monemi received the B.Sc., M.Sc., and Ph.D. degrees all in electrical and computer engineering from Shiraz University, Shiraz, Iran, and Tarbiat Modares University, Tehran, Iran, and Shiraz University, Shiraz, Iran in 2001, 2003 and 2014 respectively. After receiving his Ph.D, he has been working as the project manager in several companies, and is currently an assistant professor in the Department of Electrical Engineering, Salman Farsi University of Kazerun, Kazerun, Iran. He was a visiting researcher in the Department of Electrical and Computer Engineering, York university, Toronto, Canada from June 2019 to September 2019,  Iran. His current research interests include resource allocation in wireless networks, and traffic engineering in computer networks.
\end{IEEEbiography}

\begin{IEEEbiography}[{\includegraphics[width=1.1in,height=1.35in]{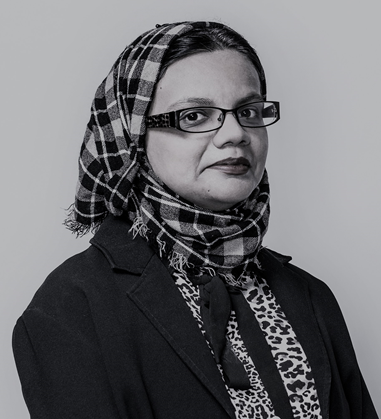}}]
{Hina Tabassum} (SM'17) Hina Tabassum is currently an Assistant Professor at the Lassonde School of Engineering, York University, Canada. Prior to that, she was a postdoctoral research associate at the Department of Electrical and Computer Engineering, University of Manitoba, Canada. She received her PhD degree from King Abdullah University of Science and Technology (KAUST). She is a Senior member of IEEE and registered Professional Engineer in the province of Ontario, Canada. She has been recognized as an Exemplary Reviewer (Top $2\%$ of all reviewers) by IEEE Transactions on Communications in 2015, 2016, 2017, and 2019. Currently, she is serving as an Associate Editor in IEEE Communications Letters and IEEE Open Journal of Communications Society. Her research interests include stochastic modeling and optimization of wireless networks including vehicular, aerial, and  satellite networks, millimeter and terahertz communication networks, software-defined networking and virtualized resource allocation in wireless networks.
\end{IEEEbiography}

\end{document}